\documentclass[11pt,letterpaper]{article}
\usepackage[T1]{fontenc}
\usepackage[utf8]{inputenc}
\usepackage{lmodern}
\usepackage[margin=1in]{geometry}
\usepackage{amsmath,amssymb,amsthm,mathtools}
\usepackage{microtype}
\usepackage{enumitem}
\usepackage{needspace}
\usepackage{booktabs}
\usepackage{tikz}
\usepackage[colorlinks=true,linkcolor=blue,citecolor=blue,urlcolor=blue]{hyperref}
\hypersetup{
  pdftitle={Wyner's Common Information and the Common Entropy Are NP-Hard to Compute},
  pdfsubject={Two polynomial-time reductions and quantitative approximation gaps}
}
\setlist[enumerate]{itemsep=3pt,topsep=5pt}
\setlist[itemize]{itemsep=3pt,topsep=5pt}
\newtheorem{theorem}{Theorem}
\newtheorem{lemma}[theorem]{Lemma}
\newtheorem{corollary}[theorem]{Corollary}
\newtheorem{proposition}[theorem]{Proposition}
\theoremstyle{definition}
\newtheorem{definition}[theorem]{Definition}
\theoremstyle{remark}
\newtheorem{remark}[theorem]{Remark}
\newcommand{\NP}{\mathrm{NP}}
\newcommand{\PT}{\mathrm{P}}
\newcommand{\R}{\mathbb{R}}
\newcommand{\cX}{\mathcal{X}}
\newcommand{\cY}{\mathcal{Y}}
\newcommand{\cW}{\mathcal{W}}
\newcommand{\cS}{\mathcal{S}}
\newcommand{\eps}{\varepsilon}
\newcommand{\qs}{q^{\star}}
\newcommand{\ind}{\mathbf{1}}
\newcommand{\indep}{\perp}
\DeclareMathOperator{\supp}{supp}
\newcommand{\CE}{\textup{\textsc{Common Entropy}}}
\newcommand{\CX}{\textup{\textsc{Cubic X3C}}}

\title{\textbf{Wyner's Common Information and the\\ Common Entropy Are NP-Hard to Compute}}
\author{ Roberto Bruno \  \
Ugo Vaccaro, \\
Department of Computer Science, University of Salerno\\
84084 Fisciano (SA), Italy\\
Email:{\tt \{rbruno,uvaccaro\}@unisa.it}}
\date{}

\begin{document}
\maketitle

\begin{abstract}
Wyner's common information $J(p)$ and the common entropy $G(p)$ of a pair of
random variables $(X,Y)$ with finite ranges and joint distribution $p$ are both
defined by an optimization over the auxiliary random variables $W$ that make $X$
and $Y$ conditionally independent: $J(p)$ is the infimum of the mutual
information $I(X,Y;W)$, and $G(p)$ is the infimum of the entropy $H(W)$. We prove
that both quantities are NP-hard to compute. For Wyner's common information,
deciding whether $J(p)$ is strictly smaller than a given integer $r$ is NP-hard,
and the problem is NP-complete on tables that are uniform on their support and
have the same number of positive entries in every row. The reduction is from
Maximum Edge Biclique and rests on an exact criterion, in terms of rectangles in
the support of $p$, for improving on the auxiliary $W=X$; every yes-instance
satisfies $J(p)<r-1/(4N^2)$ and every no-instance satisfies $J(p)=r$, where $N$
is the number of positive entries of $p$. For the common entropy, deciding
whether $G(p)\le r$ for a rational number $r$ is NP-hard, even for tables with at
most four nonzero entries in each row, all taking one of two values. The
reduction is from exact cover by $3$-sets and rests on a majorization property of
the masses of the nonnegative rank-one decompositions of $p$, combined with a
quantitative form of the Schur concavity of entropy. In both cases the separation
between yes- and no-instances is inverse-polynomial in the size of $p$, so that
approximating $J(p)$ or $G(p)$ within a suitable inverse-polynomial additive error
is also NP-hard. The two mechanisms are genuinely different: we exhibit a table on
which a large rectangle in the support lowers $J$ below the threshold of our first
reduction while leaving $G$ on it.
\end{abstract}

\medskip
\noindent\textbf{Index Terms.} Common information, common entropy, conditional
independence, NP-hardness, nonnegative rank, biclique, exact cover,
majorization.

\section{Introduction}
\label{sec:intro}

How much randomness must two agents share in order to account for the
whole of the statistical dependence between them?  Mutual information
measures the dependence of two random variables $X$ and $Y$, but it does
not answer this question: it quantifies how much the observation of one
variable says about the other, not how much common randomness is needed
in order to \emph{manufacture} the pair.  The quantity that answers the
latter question was introduced by Wyner~\cite{Wyner1975}.  For a joint
distribution $p$ of a pair $(X,Y)$ of finite valued random variables,
Wyner's common information is defined as
\begin{equation}\label{eq:intro-J}
  J(p)\;=\;\inf_{W\,:\,X\perp Y\,\mid\,W} I(X,Y;W),
\end{equation}
the infimum being taken over all auxiliary random variables $W$, jointly
distributed with $(X,Y)$, that render $X$ and $Y$ conditionally
independent.  Measuring the shared randomness by its own entropy, rather
than by the information it carries about the pair, leads to the
\emph{common entropy}
\begin{equation}\label{eq:intro-G}
  G(p)\;=\;\inf_{W\,:\,X\perp Y\,\mid\,W} H(W),
\end{equation}
introduced by Kumar, Li and El Gamal~\cite[Section~II, eq.~(1)]{KumarLiElGamal2014}.
The two infima range over the same auxiliaries, and since
$I(X,Y;W)\le H(W)$ for every $W$, one has $J(p)\le G(p)$.  This paper
studies the computational complexity of the two quantities, and shows
that both are NP-hard to compute.

Wyner gave two operational characterizations of~\eqref{eq:intro-J}.  The
first one is a corner of the rate region of the Gray--Wyner
network~\cite{GrayWyner1974}: $J(p)$ is the smallest rate that must be
carried by the common branch of the network if the total rate is to be
kept at the joint entropy $H(X,Y)$, that is, if no redundancy at all is
to be paid for the possibility of splitting the source between two private
branches.  The second one is a statement about simulation: $J(p)$ is the
smallest rate of common randomness from which two processors that do not
communicate can generate a sequence of pairs whose law is
asymptotically indistinguishable, in normalized relative entropy, from
$p^{\otimes n}$.  In this second guise common information is a precursor
of a large body of work on the approximation of output
statistics~\cite{HanVerdu1993}, on channel synthesis and soft
covering~\cite{Cuff2013}, and on the exact, rather than approximate,
generation of correlated sources~\cite{KumarLiElGamal2014}.  The last
setting is the one in which the common entropy arises: the exact common
information rate, that is, the amount of common randomness per symbol
needed to generate $(X,Y)$ exactly in a distributed manner, is the limit
of $G(p^{\otimes n})/n$ as $n\to\infty$, and it is never smaller than
Wyner's common information~\cite[Section~III]{KumarLiElGamal2014}.

Common information occupies a well understood position in the hierarchy
of dependence measures.  Writing $K(X;Y)$ for the common information of
G\'acs and K\"orner~\cite{GacsKorner1973}, that is, for the maximal rate
of randomness that the two agents can extract \emph{without} any
communication, one has
\[
  K(X;Y)\;\le\;I(X;Y)\;\le\;J(p)\;\le\;G(p)\;\le\;\min\{H(X),H(Y)\},
\]
and all the inequalities are typically strict; the two common
informations are the two extremes of the same intuition, one asking what
can be distilled from the correlation and the other what suffices to
create it.  Witsenhausen~\cite{Witsenhausen1976} computed values and
bounds for~\eqref{eq:intro-J} and obtained the variational
characterization that still underlies most of the available lower bound
techniques.  Wyner's cardinality bound $|\mathcal{W}|\le
|\mathcal{X}||\mathcal{Y}|$ turns~\eqref{eq:intro-J} into the
minimization of a continuous function over a compact set, so that the
infimum is attained; what it does not do, and this is the source of all
the difficulties, is convexify the feasible region, the set of
conditionally independent decompositions of $p$ being highly non-convex.
The same feasible region underlies~\eqref{eq:intro-G}, where in addition
the objective $H(W)$ is a concave function of the distribution of $W$,
so that computing $G(p)$ amounts to minimizing a concave function under a
conditional independence constraint~\cite[Section~II]{KumarLiElGamal2014}.

The last two decades have seen common information extended in many
 directions: to R\'enyi divergences, where the optimal rate turns
out to be remarkably robust to the choice of the distance
measure~\cite{YuTan2022}; to continuous alphabets, where a lossy source
coding interpretation is available~\cite{XuLiuChen2016}; to Gaussian
sources and to relaxed versions tailored to
estimation~\cite{SulaGastpar2021}.  It is a standard tool in the theory
of common randomness and secret key generation in
cryptography~\cite{AhlswedeCsiszar1993}, where it bounds what two
parties can securely sample~\cite{PrabhakaranPrabhakaran2014}, and it has
recently migrated into statistics and machine learning, where the
auxiliary variable $W$ of~\eqref{eq:intro-J} is read as a common latent
representation of two views of the same
object~\cite{SulaGastpar2021,HuangElGamal2024}.  Two recent accounts,
the monograph~\cite{YuTan2022} and the survey~\cite{SudanTyagiWatanabe2020},
testify to the breadth of the subject.

The common entropy has found applications of its own. In information
theory, it governs the one-shot version of exact distributed generation:
if two processors that do not communicate must output a pair with
distribution exactly $p$, using a common message encoded by a prefix-free
code, the minimum expected length $R^*$ of that message satisfies
$G(p)\le R^*<G(p)+1$~\cite[Proposition~1]{KumarLiElGamal2014}. Its
normalized limit, the exact common information rate, has since been
studied in depth: Yu and Tan~\cite{YuTan2020} showed that it coincides
with the $\infty$-R\'enyi common information and that it can be strictly
larger than Wyner's common information, settling a question raised
in~\cite{KumarLiElGamal2014}. In statistics, a decomposition of $p$ as a
mixture of product distributions is a latent class model for the pair
$(X,Y)$~\cite{Goodman1974}, and $G(p)$ is the smallest entropy of a latent
class variable that explains the whole dependence between $X$ and $Y$.
Kocaoglu, Shakkottai, Dimakis, Caramanis and
Vishwanath~\cite{KocaogluSDCV2020} have exploited this interpretation in
causal inference. Provided there are no low-entropy mediators, the common
entropy distinguishes direct causation from spurious correlation for
almost all joint distributions on simple causal graphs with two observed
variables; it can also be used to check whether a separating set found by
constraint-based causal discovery algorithms such as PC and FCI is valid,
which improves these algorithms when samples are scarce. A conditional
version of the common entropy has been used by Jiang and
Kocaoglu~\cite{JiangKocaoglu2024} to bound causal effects with
instrumental variables under weak confounding, and to falsify candidate
instruments. In these applications $G(p)$ is computed by iterative
algorithms that come with no guarantee of global
optimality~\cite{KocaogluSDCV2020}; Theorem~\ref{thm:G-main} shows that,
unless $\PT=\NP$, no polynomial-time algorithm can compute it exactly.

\subsection{Common information as a complexity measure}

Outside Shannon theory, Wyner's common information has acquired a second
life, as a lower bound technique in theoretical computer science.  The
point of contact is the nonnegative rank $\operatorname{rank}_+$.  An
\emph{extended formulation} of a polytope $P$ writes $P$ as the linear
projection of a polyhedron $Q$ in a higher-dimensional space; its size
is the number of linear inequalities that define $Q$.  The
\emph{extension complexity} of $P$ is the smallest size of an extended
formulation of $P$.  The \emph{slack matrix} of $P$ has one row for each
facet-defining inequality of $P$ and one column for each vertex of $P$,
and its entry is the slack of that inequality at that vertex; it is a
nonnegative matrix.  Yannakakis' factorization
theorem~\cite{Yannakakis1991} states that the extension complexity of
$P$ equals the nonnegative rank of its slack matrix.  Consequently, to
show that the polytope of a combinatorial optimization problem has no
extended formulation of polynomial size, it is enough to show that its
slack matrix has superpolynomial nonnegative rank.
This programme culminated in the unconditional exponential
lower bounds for the travelling salesman, cut and stable set
polytopes~\cite{Fiorini2015}, and information theory entered it with the
information complexity approach of Braverman and
Moitra~\cite{BravermanMoitra2013}.

Braun and Pokutta~\cite{BraunPokutta2016} isolated the mechanism.
Normalize a nonnegative matrix so that it is the table of a joint
distribution $p$.  Any nonnegative factorization of rank $r$ can be
normalized into a mixture of $r$ product distributions, and the index of
the mixture is an admissible auxiliary variable for both~\eqref{eq:intro-J}
and~\eqref{eq:intro-G}; since $I(X,Y;W)\le H(W)$ and the entropy of an
$r$-valued variable is at most $\log r$, one obtains
\begin{equation}\label{eq:intro-rank}
  2^{\,J(p)}\;\le\;2^{\,G(p)}\;\le\;\operatorname{rank}_+(p).
\end{equation}
Common information is thus an information-theoretic relaxation of the
logarithm of the nonnegative rank, and it is a relaxation with two
decisive advantages over the combinatorial quantity it relaxes: it
degrades gracefully under perturbations of the matrix, and entropic
estimates tensorize.  Exploiting this, Braun and Pokutta computed
essentially exactly the common information of the unique disjointness
matrix, improving on~\cite{BravermanMoitra2013}, and derived from it the
first lower bounds on the average-case and adversarial \emph{approximate}
extension complexity of polytopes, together with an information-theoretic
version of the fooling set method.  Braun, Jain, Lee and
Pokutta~\cite{BraunJainLeePokutta2017} subsequently showed that the
relaxation~\eqref{eq:intro-rank} is exact in the amortized regime:
common information is additive under tensoring and characterizes the
logarithm of the amortized nonnegative rank, that is, of the smallest
nonnegative rank attainable under tensoring and small perturbations;
they also extended Witsenhausen's dual characterization and exhibited
explicit, perturbation-robust dual witnesses.  The common entropy sits
between the two ends of~\eqref{eq:intro-rank}: a nonnegative
factorization of $p$ is charged by the entropy of its mixing weights,
whereas the nonnegative rank counts them
(Section~\ref{subsec:decompositions}).  Closely related uses of
the same circle of ideas occur in the study of the communication cost of
generating a bipartite distribution or a quantum
state~\cite{JainShiWeiZhang2013}, in the communication complexity of
correlation~\cite{Harsha2010}, in the semidefinite analogue of
Yannakakis' theorem~\cite{LeeRaghavendraSteurer2015}, and in the
decidability questions raised by non-interactive simulation of joint
distributions~\cite{GhaziKamathSudan2016}.

In this literature $J(p)$ and $G(p)$ are not information measures attached
to a source but complexity measures attached to a matrix.  It is then
natural to ask how hard they are to compute.

\subsection{The computational status of information measures}

Most of the quantities of single-letter information theory are easy to
evaluate.  Entropy and mutual information are explicit functions of the
table; the G\'acs--K\"orner common information is computed in linear
time, being the entropy of the partition of the support into connected
components; channel capacity and the rate--distortion function are
convex programs.  The combinatorial counterpart of~\eqref{eq:intro-J}
and~\eqref{eq:intro-G}, on the other hand, is hard: deciding the
nonnegative rank of a nonnegative matrix is NP-hard~\cite{Vavasis2009},
and although the problem becomes tractable when the rank is held
fixed~\cite{AroraGeKannanMoitra2012,Moitra2016}, exact nonnegative
factorizations may require irrational entries~\cite{Chistikov2017}, so
that even the format of a certificate is a delicate matter.

For Wyner's common information and for the common entropy the situation
was, until now, unresolved.  The minimization in~\eqref{eq:intro-J} is a
finite-dimensional but non-convex program; the algorithms proposed for it
are local, and come with no guarantee of global
optimality~\cite{HuangElGamal2024}.  Kumar, Li and El Gamal observe that
computing $G(p)$ is in general quite difficult, and that it is closely
related to positive, that is, nonnegative, matrix
factorization~\cite[Sections~II and~V]{KumarLiElGamal2014}.  The
NP-hardness of the nonnegative rank does not settle the question for
either quantity, because~\eqref{eq:intro-rank} is only an inequality: an
auxiliary of small mutual information or of small entropy may take many
values.  Nor is it clear a priori that the problems lie in NP, since the
optimal value of an entropic program is in general an irrational, indeed
transcendental, number, and comparisons of such numbers are not
obviously certifiable.

\subsection{Results of this paper}

We settle the complexity of both quantities.  Both proofs rest on the same
structural fact: an admissible auxiliary $W$ is the same thing as a
decomposition of $p$ into nonnegative matrices of rank one, the
distribution of $W$ being the vector of masses of the components, and
each component is supported on a rectangle contained in the support of $p$
(Lemma~\ref{lem:decomp}).  The two quantities weigh the components
differently, and this leads to two different combinatorial problems.

\emph{Wyner's common information.}  Our first main result,
Theorem~\ref{thm:J-main}, states that deciding whether $J(p)<r$, for a
rational probability table $p$ and an integer $r$ given in binary, is
NP-hard under deterministic polynomial-time many-one reductions.
Hardness does not depend on any numerical fine tuning: it persists on
the class of tables that are uniform on their support, have $n=2^{r}$
rows, exactly $d$ positive entries in every row, every positive entry
equal to $1/(nd)$, and at most $d$ positive entries in every column.  On
this restricted class the decision problem is moreover NP-complete.

The proof rests on an exact combinatorial criterion for strict
improvement over the trivial auxiliary variable $W=X$.  For a table $p$
uniform on a support $E\subseteq\mathcal{X}\times\mathcal{Y}$ with $d$
positive entries per row, Corollary~\ref{cor:criterion} states that
\[
  J(p)<\log n
  \qquad\Longleftrightarrow\qquad
  E\ \text{contains a rectangle } A\times B \text{ with } |A||B|>d,
\]
and that $J(p)=\log n$ exactly when no such rectangle exists.  The
information-theoretic question is thereby reduced to a question about
bicliques of the bipartite graph carrying the support, and NP-hardness
follows by a direct reduction from Maximum Edge Biclique, which is
NP-complete by a theorem of Peeters~\cite{Peeters2003}.

Two features of this criterion should be pointed out.  First, the hard
instances are uniform tables, so that the difficulty is created entirely
by the combinatorics of the support -- precisely the rectangle structure
that governs nonnegative rank and communication complexity lower bounds.
Second, the criterion supplies certificates that are purely
combinatorial, namely a pair of index sets; combined with the choice
$n=2^{r}$, which makes the threshold $\log n=r$ an integer, this allows
the comparison of an entropic quantity with the threshold to be verified
without any real arithmetic, and is what yields membership in NP on the
restricted class in spite of the transcendence phenomena recalled above.

Hardness is finally shown to be numerically robust.  Proposition~\ref{prop:gap}
quantifies the dichotomy: every yes-instance produced by the reduction
satisfies $J(p)<r-1/(4N^{2})$, while every no-instance satisfies
$J(p)=r$, where $N$ is the number of positive entries of the table.  The
gap is inverse polynomial in the input length, and Theorem~\ref{thm:approximation} converts it
into hardness of numerical approximation: producing a rational $q$ with
$|q-J(p)|\le 1/(16N^{2})$ is NP-hard, and consequently, unless
$\mathrm{P}=\mathrm{NP}$, no algorithm can return $2^{-s}$-approximations
of $J(p)$ in time polynomial in the length of the table and in $s$ given
in unary.  The obstruction is already present at the very modest
precision $s=O(\log N)$.

\emph{The common entropy.}  Our second main result,
Theorem~\ref{thm:G-main}, states that deciding whether $G(p)\le r$, for a
rational probability table $p$ and a rational number $r$, is NP-hard; this
remains true for tables with at most four nonzero entries in each row, all
nonzero entries taking one of two values.  Here only the masses of the
components of a decomposition matter, and the decisive combinatorial
question is how much mass a single rectangle can carry.  The reduction is
from exact cover by $3$-sets, in the cubic version shown NP-complete by
Moore and Robson~\cite{MooreRobson2001}.  For the tables we construct, the
mass vector of every rank-one decomposition is majorized by an explicit
two-level probability vector $\qs$, which is attained if and only if the
underlying instance has an exact cover.  By the Schur concavity of
entropy, $G(p)=H(\qs)$ on yes-instances, and a quantitative form of Schur
concavity shows that, on no-instances, $G(p)$ exceeds $H(\qs)$ by an amount
that is inverse-polynomial in the size of $p$.  This gap allows the value
$H(\qs)$, which is in general irrational, to be replaced by a rational
threshold, and, as for Wyner's common information, it makes the
approximation of $G(p)$ within an inverse-polynomial additive error NP-hard
(Corollary~\ref{cor:G-approx}).

\emph{Two different mechanisms.}  Wyner's common information rewards
components that spread their mass uniformly over large rectangles, and is
therefore governed by the largest rectangle in the support; the common
entropy rewards components of large mass, and is governed by how much of
the mass of $p$ can be concentrated on few rectangles.  The difference is
intrinsic: Proposition~\ref{prop:J-vs-G} exhibits a table in the
restricted class of Theorem~\ref{thm:J-main} on which a large rectangle
lowers $J$ below the threshold of our first reduction while $G$ stays on
it.  Thus neither hardness result is a consequence of the other, and two
different reductions are needed.

\subsection{Intractability in Shannon theory}

Results of this kind have a tradition in information theory.  Berlekamp,
McEliece and van Tilborg~\cite{Berlekamp1978} showed that decoding a
general linear code is NP-hard, and Vardy~\cite{Vardy1997} showed that
computing the minimum distance of a code is also NP-hard.
Our paper belongs to the research line on understanding the computational complexity
of optimization problems in Shannon Theory (e.g., \cite{BrunoVaccaroCMEIA, noi, CGV,
ebrahimi2024minimum,
SC, Kov} and references therein quoted).  Among these problems, the
closest relative of the common entropy is the minimum-entropy coupling
problem, another entropy minimization over a set of joint distributions,
which is NP-hard~\cite[Theorem~3.4]{Kov}.
Each of these results marks the boundary beyond which an information theoretic
quantity, defined
by an optimization over an infinite family of objects, ceases to be an
effectively computable function of its data.  The present paper places
Wyner's common information and the common entropy on the hard side of
that boundary.

\subsection{Organization of the paper}

Section~\ref{sec:statement} establishes the notation, the computational
model and the statements of the two main theorems, together with the
correspondence between admissible auxiliaries and nonnegative rank-one
decompositions on which both proofs rest.

Sections~\ref{sec:rectangles}--\ref{sec:gap} concern Wyner's common
information.  Section~\ref{sec:rectangles}
characterizes, for tables that are uniform on their support, when the
trivial auxiliary variable $W=X$ can be strictly improved, and
establishes the rectangle criterion.  Section~\ref{sec:reduction} turns that criterion
into a polynomial-time many-one reduction from Maximum Edge Biclique and
proves Theorem~\ref{thm:J-main}, together with membership in NP on the
restricted class.  Section~\ref{sec:gap} reuses the entropy estimate of Section~\ref{sec:rectangles} to
exhibit an explicit gap between yes-instances and no-instances, and
deduces the hardness of approximating $J(p)$ numerically.

Sections~\ref{sec:x3c}--\ref{sec:G-proof} concern the common entropy.
Section~\ref{sec:x3c} outlines the proof of Theorem~\ref{thm:G-main} and
introduces the restricted exact cover problem from which we reduce.
Section~\ref{sec:matrix} constructs the reduction table and analyses its
rank-one decompositions, Section~\ref{sec:entropy-gap} proves the entropy
gap, and Section~\ref{sec:G-proof} completes the proof and derives the
hardness of approximating $G(p)$.

Section~\ref{sec:comparison} compares the two results and lists some open
problems.

\section{Preliminaries}
\label{sec:statement}

\subsection{Notation}

All logarithms are to base two; $\ln$ denotes the natural logarithm, and we posit
$0\log 0=0$; $[n]=\{1,\dots,n\}$ and $x_+=\max\{x,0\}$. A \emph{probability
vector} is a family $q=(q_i)_{i\in I}$ of nonnegative real numbers, indexed by a
finite set $I$, with $\sum_{i\in I}q_i=1$. Its entropy is
$H(q)=-\sum_{i\in I}q_i\log q_i$; the entropy $H(W)$ of a random variable $W$
with finitely many values is the entropy of its distribution.

Let $p=(p(x,y))_{x\in \mathcal X, y\in \mathcal Y}$ be the joint probability distribution  of random variables
$X, Y$ taking values in the
 finite alphabets $\mathcal X,\mathcal Y$. We see $p=(p(x,y))_{x\in \mathcal X, y\in \mathcal Y}$ as a table, or matrix, where
its rows are indexed by elements $x$ in
$\mathcal X$ and its columns by elements $y$ in $\mathcal Y$.

\subsection{Admissible auxiliaries, common information and common entropy}

An \emph{auxiliary} is a random variable $W$ defined jointly with
$(X,Y)$ so that the  marginal of  $(X,Y)$ remains $p$.
As in~\cite{Wyner1975,KumarLiElGamal2014}, auxiliaries take finitely many values.
$W$ is  \emph{admissible} when
it holds that
\begin{equation}
\Pr(X=x,Y=y\mid W=w)
=\Pr(X=x\mid W=w)\Pr(Y=y\mid W=w)
\label{eq:independence}
\end{equation}
for all $x,y$ and all $w$ of positive probability, that is, when $X$ and $Y$ are
conditionally independent given $W$, written $X\perp Y\mid W$.
Wyner's common information~\cite{Wyner1975} and the common
entropy~\cite{KumarLiElGamal2014} of $p$ are
\begin{equation}
J(p)=\inf_{W:\,X\perp Y\mid W} I(X,Y;W)
\label{eq:wyner}
\end{equation}
and
\begin{equation}\label{eq:G}
  G(p)=\inf_{W:\,X\perp Y\mid W} H(W),
\end{equation}
the infima being taken over all admissible auxiliaries.
We use the  well known identities
\begin{equation}
\begin{split}
I(X,Y;W)&=H(X,Y)-H(X,Y\mid W),\\
H(X,Y\mid W)&=\sum_{w:\,\Pr(W=w)>0}
  \Pr(W=w)H\bigl(p(\cdot,\cdot\mid w)\bigr).
\end{split}
\label{eq:conditional}
\end{equation}
In particular, $W=X$ is admissible.
Thus
\begin{equation}
I(X,Y;X)=H(X),\quad  \mbox { and } \quad J(p)\le H(X).
\label{eq:row-upper}
\end{equation}
Finding a better auxiliary for $J$ therefore amounts to making
$H(X,Y\mid W)$ strictly larger than $H(Y\mid X)$.
Since $I(X,Y;W)=H(W)-H(W\mid X,Y)\le H(W)$ for every auxiliary, and $W=X$ and
$W=Y$ are admissible, we also have
\begin{equation}
J(p)\;\le\;G(p)\;\le\;\min\{H(X),H(Y)\}.
\label{eq:J-le-G}
\end{equation}

\subsection{Computational model and main results}

We use the standard deterministic Turing bit model. A probability table $p$ is
given explicitly, including its zero entries, with rational numerators and
positive denominators encoded in binary. Running time is measured against
the total bit length of the input. Since $J(p)$ and $G(p)$ are in general
irrational, we formulate their computation as threshold problems.

To prove that computing $J(p)$ is
NP-hard, we prove the NP-hardness of an associated decision problem.
More precisely, the decision problem we will study has input $(p,r)$,
where $r$ is a nonnegative integer encoded in binary, and asks whether it holds that
\begin{equation}
J(p)<r.
\label{eq:decision}
\end{equation}
The answer is no when $J(p)=r$.

\begin{theorem}[Hardness of Wyner's common information]
\label{thm:J-main}
Deciding~\eqref{eq:decision} is NP-hard under deterministic polynomial-time
many-one reductions. NP-hardness remains true under all of the following restrictions
on the table $p$:
\begin{enumerate}
\item the number of rows is $n=2^r$, with $n\ge2$;
\item every row has exactly $d\ge1$ positive entries;
\item every positive entry equals $1/(nd)$;
\item every column has at most $d$ positive entries.
\end{enumerate}
Additionally, on this restricted class, the decision problem is NP-complete.
\end{theorem}

For the common entropy we use a rational threshold and a non-strict
inequality.

\begin{definition}\label{def:problem}
\CE: given a rational probability table $p$ and a rational number $r$, decide whether
$G(p)\le r$.
\end{definition}

\begin{theorem}[Hardness of the common entropy]\label{thm:G-main}
\CE\ is $\NP$-hard under polynomial-time many-one reductions. This remains true for instances in
which every row of $p$ has at most four nonzero entries and all nonzero entries of $p$ take one of
two values.
\end{theorem}

Consequently, a polynomial-time algorithm that computes $G(p)$ exactly, in any output format that
allows comparison with a given rational number in polynomial time, would decide \CE\ in polynomial
time; the same holds for $J(p)$ and the decision problem~\eqref{eq:decision}.
The two problems are formulated differently for a reason. On the hard
instances for $J$, the no-instances satisfy $J(p)=r$ exactly, at an integer
threshold (Remark~\ref{rem:strict}). On the hard instances for $G$, the
yes-instances satisfy $G(p)=h_0$ for an explicit, in general irrational,
number $h_0$, and the no-instances satisfy $G(p)\ge h_0+\Delta$; a rational
threshold is placed inside this gap (Section~\ref{sec:G-proof}).

\subsection{Rank-one decompositions}
\label{subsec:decompositions}

All matrices have real entries, and their rows and columns are indexed by finite sets; inequalities
between matrices are meant entrywise. The \emph{support} of a matrix $F$ is
$\supp F=\{(x,y)\colon F(x,y)\ne0\}$ and its \emph{mass} is $\mu(F)=\sum_{x,y}F(x,y)$.
A \emph{nonnegative rank-one matrix} is a matrix of the form $F=gh^{\top}$, where $g$ and $h$ are
nonzero vectors with nonnegative entries. (Every nonzero nonnegative matrix of rank one has this
form, but we shall not need this fact.) Its support is
\begin{equation}\label{eq:rectangle}
  \supp F=\supp g\times\supp h,\quad\text{a rectangle.}
\end{equation}

\begin{definition}\label{def:decomposition}
A \emph{rank-one decomposition} of a nonnegative matrix $A$ is a finite family $(F_w)_{w\in\cW}$ of
nonnegative rank-one matrices such that $\sum_{w\in\cW}F_w=A$. The matrices $F_w$ are its
\emph{components}, and $(\mu(F_w))_{w\in\cW}$ is its \emph{mass vector}.
\end{definition}

Since all terms are nonnegative, every component of a rank-one decomposition of $A$ satisfies
$0\le F_w\le A$; in particular, its support is a rectangle contained in the support of $A$.

\begin{lemma}\label{lem:decomp}
Let $p$ be a probability table on $\cX\times\cY$.
\begin{enumerate}[label=\textup{(\alph*)}]
\item If $W$ is an admissible auxiliary for $p$, then the matrices
$F_w(x,y)=\Pr(X=x,\,Y=y,\,W=w)$, for $w$ with $\Pr(W=w)>0$, form a rank-one decomposition of $p$
with $\mu(F_w)=\Pr(W=w)$.
\item Conversely, for every rank-one decomposition $(F_w)_{w\in\cW}$ of $p$ there is an admissible
auxiliary $W$ with $\Pr(W=w)=\mu(F_w)$ for every $w\in\cW$.
\end{enumerate}
\end{lemma}

\begin{proof}
(a) For each $w$ with $\Pr(W=w)>0$, conditional independence gives $F_w=g_wh_w^{\top}$ with
$g_w(x)=\Pr(W=w)\Pr(X=x\mid W=w)$ and $h_w(y)=\Pr(Y=y\mid W=w)$, which are nonzero nonnegative
vectors. Thus $F_w$ is a nonnegative rank-one matrix with $\mu(F_w)=\Pr(W=w)$, and
$\sum_wF_w=p$ by the law of total probability.

(b) Write $F_w=g_wh_w^{\top}$ with nonnegative vectors $g_w,h_w$, scaled so that the entries of
$h_w$ sum to one; then the entries of $g_w$ sum to $\mu(F_w)>0$. The function
$(x,y,w)\mapsto F_w(x,y)$ is a probability distribution on $\cX\times\cY\times\cW$, because
$\sum_w\mu(F_w)=\mu(p)=1$. Let $(X,Y,W)$ have this distribution. The pair $(X,Y)$ has distribution
$\sum_wF_w=p$, the variable $W$ has distribution $(\mu(F_w))_{w}$, and conditionally on $W=w$ the
pair $(X,Y)$ has distribution $F_w/\mu(F_w)=(g_w/\mu(F_w))h_w^{\top}$, which is a product of a
distribution on $\cX$ and a distribution on $\cY$. Hence $X\indep Y\mid W$.
\end{proof}

By Lemma~\ref{lem:decomp}, the least number of values of an admissible auxiliary for $p$ is the
nonnegative rank of $p$, that is, the least number of nonnegative rank-one matrices summing to $p$;
hence $G(p)\le\log\operatorname{rank}_+(p)$, as recorded in~\eqref{eq:intro-rank}. The proof of
Lemma~\ref{lem:no-rectangle} in Section~\ref{sec:rectangles} rests on the same observation: the
conditional tables $\Pr(X=x,Y=y\mid W=w)$ are products, so their supports are rectangles contained
in the support of $p$.

\subsection{Majorization}\label{subsec:majorization}

Entropy does not depend on the order of the entries of a probability vector, and neither does the
notion of majorization defined below. When dealing with majorization we therefore list the entries
of a probability vector $q$ with $\lvert I\rvert$ entries in nonincreasing order,
$q_1\ge q_2\ge\dots\ge q_{\lvert I\rvert}$, and we put $q_k=0$ for $k>\lvert I\rvert$; in this way,
vectors with different numbers of entries can be compared.

\begin{definition}\label{def:majorization}
For probability vectors $q$ and $q'$, we say that $q$ is \emph{majorized} by $q'$, and write
$q\preceq q'$, if
\[
  q_1+\dots+q_k\le q'_1+\dots+q'_k\qquad\text{for every } k\ge 1.
\]
\end{definition}

Since both vectors sum to one, this is the usual notion of
majorization~\cite[Chapter~1]{MarshallOlkinArnold2011}, applied after padding the shorter vector
with zeros. It is well known that Shannon  entropy is Schur
concave, that is,  $q\preceq q'$ implies $H(q)\ge H(q')$; see~\cite{MarshallOlkinArnold2011}.
Lemma~\ref{lem:entropy} below proves a quantitative version of this fact for the vectors $q'$
that arise in the reduction for the common entropy.

\section{Hardness of Wyner common information}
\label{sec:rectangles}

Sections~\ref{sec:rectangles}--\ref{sec:gap} are devoted to Wyner's common information and
to the proof of Theorem~\ref{thm:J-main}. Our approach is the following.
We first characterize strict improvement over the auxiliary $W=X$ in terms
of rectangles in the support of $p$. We then reduce Maximum Edge Biclique to this
criterion, using Peeters's NP-completeness theorem~\cite[Theorem~1]{Peeters2003}. 
Section~\ref{sec:gap} reuses the same entropy estimate
to obtain an explicit gap and hardness of rational approximation of \eqref{eq:wyner}.

Let $E\subseteq\mathcal X\times\mathcal Y$ be nonempty. For each
$x\in\mathcal X$, define
\[
E_x=\{y\in\mathcal Y:(x,y)\in E\}.
\]
Assume $|\mathcal X|=n$ and $|E_x|=d\ge1$ for every $x$.
Thus the row indexed by $x$ contains the $d$  pairs
$\{x\}\times E_x$, and $N:=|E|=nd$.
A \emph{rectangle} in $E$ is a set $A\times B\subseteq E$ with
$A,B$ nonempty; the size of the rectangle $A\times B$  is $|A||B|$.

Consider the uniform probability table on $E$, that is, for all $x\in \mathcal X$ and $y\in \mathcal Y$
\begin{equation}
p(x,y)=
\begin{cases}
1/N,&(x,y)\in E,\\
0,&(x,y)\notin E.
\end{cases}
\label{eq:uniform}
\end{equation}
For every $x$, summing the  $d$ positive entries in $E_x$ gives
$\Pr(X=x)=d/N=1/n$. Consequently,
\[
\Pr(Y=y\mid X=x)=
\begin{cases}
(1/N)/(1/n)=1/d,&y\in E_x,\\
0,&y\notin E_x.
\end{cases}
\]
Hence $X$ is uniform on $\mathcal X$, and the conditional law of $Y$ given
$X=x$ is uniform on $E_x$. In particular,
\begin{equation}
H(X)=\log n,\qquad H(Y\mid X)=\log d,\qquad H(X,Y)=\log N.
\label{eq:entropies}
\end{equation}

\Needspace{7\baselineskip}
\begin{lemma}
\label{lem:no-rectangle}
If every rectangle in $E$ has at most $d$ entries, then the table
\eqref{eq:uniform} satisfies $J(p)=\log n$.
\end{lemma}

\begin{proof}
Fix an admissible $W$. For each $w$ with $\Pr(W=w)>0$, let
$$q_w(x,y)=\Pr(X=x,Y=y\mid W=w)$$ and define
\[
A_w=\{x:\Pr(X=x\mid W=w)>0\},
\qquad
B_w=\{y:\Pr(Y=y\mid W=w)>0\}.
\]
Since 
$$q_w(x,y)=\Pr(X=x,Y=y\mid W=w)=\Pr(X=x\mid W=w)\Pr(Y=y\mid W=w), $$
 the support of $q_w$ is exactly
$A_w\times B_w$. Moreover, one has that
\[
p(x,y)=\sum_{v:\,\Pr(W=v)>0}\Pr(W=v)q_v(x,y)
\ge \Pr(W=w)q_w(x,y).
\]
Since all summands are nonnegative, it holds that $q_w$ is zero wherever $p$ is zero.
By \eqref{eq:uniform}, 
 $A_w\times B_w\subseteq E$. Therefore, by the hypothesis of the lemma, it holds that  $|A_w||B_w|\le d$.

Recall that a probability vector with $s$ positive entries has entropy at most $\log s$.
Applying this inequality to each joint conditional distribution $q_w$ gives
$H(X,Y\mid W=w)\le\log d$. Averaging over $w$ yields
$H(X,Y\mid W)\le\log d$, and therefore
\[
I(X,Y;W)\ge\log N-\log d=\log n.
\]
This holds for every admissible $W$, while $W=X$ attains $\log n$
by~\eqref{eq:row-upper} and~\eqref{eq:entropies}.
\end{proof}

\begin{lemma}
\label{lem:large-rectangle}
If $E$ contains a rectangle with more than $d$ entries, then the table
\eqref{eq:uniform} satisfies $J(p)<\log n$.
\end{lemma}

\begin{proof}
Choose $A\times B\subseteq E$ and write
$a=|A|$, $b=|B|$, and $m=ab>d$.
Since $B\subseteq E_u$ for every $u\in A$, we have $b\le d$.
For $0\le t<1$, define he random variable $W_t$ that takes  values in
$\mathcal X\cup\{\ast\}$, where $\ast\notin\mathcal X$, as follows.
Draw $(X,Y)$ from $p$. If $(X,Y)\in A\times B$, output $\ast$ with
probability $t$ and the row label $X$ with probability $1-t$;
otherwise output $X$. Explicitly, the resulting joint distribution is
\begin{equation}
\Pr(X=x,Y=y,W_t=w)=
\begin{cases}
t/N,&(x,y)\in A\times B,\quad w=\ast,\\
(1-t)/N,&(x,y)\in A\times B,\quad w=x,\\
1/N,&(x,y)\in E\setminus(A\times B),\quad w=x,\\
0,&\text{otherwise}.
\end{cases}
\label{eq:joint-construction}
\end{equation}
One can easily see that summing over $w$ gives $p(x,y)$.
In particular, $W_0=X$.

\smallskip
We now compute the conditional distributions and check for the admissibility of $W_t.$
For $t>0$, the  label $\ast$ has probability
\[
\Pr(W_t=\ast)=\sum_{(x,y)\in A\times B}\frac{t}{N}
=\frac{mt}{N}.
\]
Let $\mathbf1_C$ denote the indicator function of a set $C$.
We first compute the conditional joint distribution given
$W_t=\ast$. For $t>0$, the construction and $m=ab$ give
\[
\begin{aligned}
\Pr(X=x,Y=y\mid W_t=\ast)
&=\frac{\Pr(X=x,Y=y,W_t=\ast)}
        {\Pr(W_t=\ast)}\\
&=\frac{(t/N)\mathbf1_A(x)\mathbf1_B(y)}{mt/N}\\
&=\frac{\mathbf1_A(x)\mathbf1_B(y)}{ab}.
\end{aligned}
\]
Thus this conditional joint distribution is uniform on
$A\times B$ and zero outside it.

We next compute its two conditional marginals.
Summing over $y$ gives
\[
\begin{aligned}
\Pr(X=x\mid W_t=\ast)
&=\sum_{y\in\mathcal Y}
  \Pr(X=x,Y=y\mid W_t=\ast)\\
&=\frac{\mathbf1_A(x)}{ab}
  \sum_{y\in\mathcal Y}\mathbf1_B(y)
=\frac{\mathbf1_A(x)}a,
\end{aligned}
\]
because $|B|=b$. Similarly, since $|A|=a$,
\[
\begin{aligned}
\Pr(Y=y\mid W_t=\ast)
&=\sum_{x\in\mathcal X}
  \Pr(X=x,Y=y\mid W_t=\ast)\\
&=\frac{\mathbf1_B(y)}{ab}
  \sum_{x\in\mathcal X}\mathbf1_A(x)
=\frac{\mathbf1_B(y)}b.
\end{aligned}
\]
Hence, for every $x\in\mathcal X$ and
$y\in\mathcal Y$,
\begin{equation}
\begin{aligned}
\Pr(X=x,Y=y\mid W_t=\ast)
&=\frac{\mathbf1_A(x)\mathbf1_B(y)}{ab}\\
&=\Pr(X=x\mid W_t=\ast)
  \Pr(Y=y\mid W_t=\ast).
\end{aligned}
\label{eq:star-product}
\end{equation}
For a row label $u\in\mathcal X$, summing the masses that retain this label
gives
\[
\Pr(W_t=u)=
\begin{cases}
[b(1-t)+(d-b)]/N=(d-bt)/N,&u\in A,\\
d/N,&u\notin A.
\end{cases}
\]
All these probabilities are positive, since
$d-bt\ge d(1-t)>0$, including when $b=d$.
Write $q_{u,t}(y)=\Pr(Y=y\mid W_t=u)$.
If $u\notin A$, then $q_{u,t}(y)=\mathbf1_{E_u}(y)/d$.
If $u\in A$, division by $(d-bt)/N$ gives
\begin{equation}
q_{u,t}(y)=
\begin{cases}
(1-t)/(d-bt),&y\in B,\\
1/(d-bt),&y\in E_u\setminus B,\\
0,&y\notin E_u.
\end{cases}
\label{eq:affected-row}
\end{equation}
The event $W_t=u$ implies $X=u$. Thus, for every $x,y$,
\[
\Pr(X=x,Y=y,W_t=u) =\Pr(X=x)\Pr(Y=y,W_t=u\mid X=x)
=
\mathbf1_{\{u\}}(x)\Pr(Y=y,W_t=u).
\]
In fact, if  $x\ne u$, both sides are zero; if $x=u$, the two
events whose probabilities appear are identical.
Dividing by $\Pr(W_t=u)>0$ gives
\[
\Pr(X=x,Y=y\mid W_t=u)
=
\mathbf1_{\{u\}}(x)\Pr(Y=y\mid W_t=u)
=
\mathbf1_{\{u\}}(x)q_{u,t}(y).
\]
Moreover, since $X=u$ is certain given that  $W_t=u$, it holds that
\[
\Pr(X=x\mid W_t=u)=\mathbf1_{\{u\}}(x).
\]
Consequently,
\begin{equation}
\begin{split}
\Pr(X=x,Y=y\mid W_t=u)
&=\mathbf1_{\{u\}}(x)q_{u,t}(y)\\
&=\Pr(X=x\mid W_t=u)\Pr(Y=y\mid W_t=u).
\end{split}
\label{eq:row-product}
\end{equation}
Equations~\eqref{eq:star-product} and~\eqref{eq:row-product} prove the
required factorization for every label of positive probability. 

At \(t=0\), the construction gives \(W_0=X\) with probability one. The label \(\ast\) never occurs, so no conditional distribution given that label is needed. Every row label \(u\) has probability \(d/N>0\), and the factorization~\eqref{eq:row-product} remains valid at \(t=0\). Hence \(W_0\) is admissible as well. Thus $W_t$ is admissible for every $0\le t<1$.

\smallskip

In formula~\eqref{eq:conditional}, each label $w$ contributes its probability multiplied
by the entropy of its joint conditional distribution. The common label $\ast$  is
uniform on $m$ pairs and contributes $(mt/N)\log m$ (zero at $t=0$).
Each of the $n-a$ row labels outside $A$ contributes $(d/N)\log d$.
For $u\in A$, the $b$ probabilities of the first type
in~\eqref{eq:affected-row} and the $d-b$ of the second type give
\begin{align*}
&\Pr(W_t=u)H(X,Y\mid W_t=u)\\
&\quad=-\frac1N\left[
 b(1-t)\log\frac{1-t}{d-bt}
 +(d-b)\log\frac1{d-bt}\right]\\
&\quad=\frac1N\left[
 (d-bt)\log(d-bt)-b(1-t)\log(1-t)\right].
\end{align*}
This expression is the same for every $u\in A$, regardless of the names of
the elements of $E_u\setminus B$.

There are $a$ affected row labels. Adding their contributions
to those of the  label $\ast$ and the $n-a$ unaffected row
labels, and using $ab=m$, gives
\[
\begin{aligned}
N H(X,Y\mid W_t)
={}&mt\log m+(n-a)d\log d\\
&+a(d-bt)\log(d-bt)
-m(1-t)\log(1-t).
\end{aligned}
\]
To isolate the contribution coming  from $H(X,Y\mid W_0)=\log d$, use
$N=nd$ to write
\[
(n-a)d\log d=N\log d-ad\log d.
\]
Substituting this identity into the preceding expression
and subtracting $N\log d$ yields
\[
\begin{aligned}
N\bigl(H(X,Y\mid W_t)-\log d\bigr)
={}&mt\log m\\
&+a\bigl[(d-bt)\log(d-bt)-d\log d\bigr]\\
&-m(1-t)\log(1-t).
\end{aligned}
\]
Denote the right-hand side of the above formula by $\Phi(t)$. Dividing by $N$
gives the conditional-entropy identity

$$
H(X,Y\mid W_t)
=\log d+\frac{\Phi(t)}N.
$$
 Since
$H(X,Y)=\log N$ and $\log N-\log d=\log n$, the next 
mutual-information identity follows:
\begin{equation}
\begin{aligned}
I(X,Y;W_t)
&=H(X,Y)-H(X,Y\mid W_t)\\
&=\log N-\log d-\frac{\Phi(t)}N
 =\log n-\frac{\Phi(t)}N,
\end{aligned}
\label{eq:gain-identity}
\end{equation}
where
\begin{equation}
\begin{split}
\Phi(t)={}&mt\log m
+a\bigl[(d-bt)\log(d-bt)-d\log d\bigr]\\
&\ -m(1-t)\log(1-t).
\end{split}
\label{eq:gain-function}
\end{equation}
Thus $\Phi(t)/N$ is the increase in conditional entropy
over $\log d$, and equally the decrease in mutual
information below $\log n$. At $t=0$, these identities
hold with $\Phi(0)=0$ and $W_0=X$.

\smallskip
We now claim that, for $0\le t<1$,
\begin{equation}
\Phi(t)\ge mt\left(\log\frac md-\frac{t}{\ln2}\right).
\label{eq:gain-bound}
\end{equation}
We use the well known inequality
$\ln u\le u-1$ for $u>0$.
Applying it to $u=d/z$, with $z>0$, yields
$z\log(z/d)\ge(z-d)/\ln2$ and hence
\[
z\log z-d\log d
=(z-d)\log d+z\log(z/d)
\ge(z-d)\left(\log d+\frac1{\ln2}\right).
\]
Taking $z=d-bt>0$ gives
\[
(d-bt)\log(d-bt)-d\log d
\ge-bt\left(\log d+\frac1{\ln2}\right).
\]
Also, $\ln(1-t)\le-t$, so
\[
-(1-t)\log(1-t)\ge\frac{t(1-t)}{\ln2}.
\]
Substituting these inequalities in~\eqref{eq:gain-function} and using
$ab=m$, one gets 
\[
\Phi(t)\ge mt\log m-mt\left(\log d+\frac1{\ln2}\right)
       +\frac{mt(1-t)}{\ln2}
=mt\left(\log\frac md-\frac{t}{\ln2}\right),
\]
proving ~\eqref{eq:gain-bound}.
Since $m>d$, this lower bound is positive whenever
$0<t<\min\{1,\ln(m/d)\}$.
For any such $t$, admissibility and~\eqref{eq:gain-identity} imply
\[
J(p)\le I(X,Y;W_t)=\log n-\frac{\Phi(t)}N<\log n.\qedhere
\]
\end{proof}

It is  useful to summarize the results of Lemmas~\ref{lem:no-rectangle}
and~\ref{lem:large-rectangle} in the following Corollary.
\begin{corollary}
\label{cor:criterion}
For the uniform table~\eqref{eq:uniform} with $d$ positive entries in each row,
\begin{equation}
J(p)<\log n
\quad\Longleftrightarrow\quad
E\text{ contains a rectangle }A\times B\text{ with }|A||B|>d.
\label{eq:criterion}
\end{equation}
If no such rectangle exists, $J(p)=\log n$.
\end{corollary}


\section{The polynomial-time reduction}
\label{sec:reduction}

We prove the NP-hardness of the decision problem \eqref{eq:decision} by reducing from  
\emph{Maximum Edge Biclique}.
The {Maximum Edge Biclique} decision problem has as input an explicitly
given finite bipartite graph $\Gamma=(L,R,E_\Gamma)$ and a positive integer $k$ in binary.
It asks whether there are nonempty sets $A\subseteq L$, $B\subseteq R$ with
\[
A\times B\subseteq E_\Gamma,\qquad |A||B|\ge k.
\]
The {Maximum Edge Biclique} decision problem is  NP-complete by
Peeters~\cite[Theorem~1, p.~652]{Peeters2003}. 

\begin{lemma}
\label{lem:reduction}
Given an instance $(\Gamma,k)$ of Maximum Edge Biclique,
where $\Gamma=(L,R,E_\Gamma)$ is an explicitly given finite bipartite
graph and $k\ge1$ is encoded in binary, one can construct,
in polynomial time, integers $d\ge1$, $r\ge1$, and a
bipartite graph
\[
\Gamma'=(\mathcal X,\mathcal Y,E),
\qquad E\subseteq\mathcal X\times\mathcal Y,
\]
with disjoint left and right vertex sets $\mathcal X$
and $\mathcal Y$, respectively, such that:
\begin{enumerate}
\item $|\mathcal X|=n=2^r$, every vertex of $\mathcal X$
has degree $d$, and every vertex of $\mathcal Y$ has
degree between $1$ and $d$;
\item $\Gamma$ has a biclique with at least $k$ edges if and
only if there exist nonempty sets $A\subseteq\mathcal X$
and $B\subseteq\mathcal Y$ such that
$A\times B\subseteq E$ and $|A||B|>d$;
\item the complete rational probability table indexed by
$\mathcal X\times\mathcal Y$ and uniform on $E$, together
with the integer $r$, has polynomial encoding length.
\end{enumerate}
\end{lemma}

\begin{proof}
Let $e=|E_\Gamma|$. We call a vertex in $R$  \emph{private} if it has exactly one
neighbor. We will use the following observation: in a graph whose left
degrees are all equal to $d$, a rectangle containing a private right vertex has
only one left vertex, and therefore at most $d$ edges.
Indeed, suppose a rectangle $A\times B$ contains a private
right vertex $y_0$, whose unique neighbor is $x_0$.
Every vertex of $A$ must be adjacent to $y_0$, so
$A=\{x_0\}$, since $A$ is nonempty.
Moreover, every vertex of $B$ must be adjacent to $x_0$.
As $x_0$ has degree $d$, we obtain
$|A||B|=|B|\le d$.

\smallskip
\noindent\emph{Step 1}
If $k>e$ (in particular, if $e=0$), then $(\Gamma,k)$ is a
no-instance of Maximum Edge Biclique.
In this case, construct $\Gamma'=(\mathcal X,\mathcal Y,E)$ by
setting
\[
\mathcal X=\{x_1,x_2\},\qquad
\mathcal Y=\{y_1,y_2\},\qquad
E=\{(x_1,y_1),(x_2,y_2)\},
\]
where the four vertices are distinct, and put
$d=1$, $r=1$, and $n=2$.

We verify the required conditions.
For condition~1, we have $|\mathcal X|=2=2^r$,
and every left and right vertex has degree $1=d$.

For condition~2, the original graph $\Gamma$ has no biclique
with at least $k$ edges. We show that $E$ likewise has
no rectangle with more than $d$ entries.
Let $A\times B\subseteq E$ be a nonempty rectangle.
The set $A$ cannot contain both $x_1$ and $x_2$,
because these vertices have no common neighbor and
$B$ is nonempty. Hence $A=\{x_i\}$ for some
$i\in\{1,2\}$. Since $y_i$ is the unique neighbor
of $x_i$, necessarily $B=\{y_i\}$. Thus
$|A||B|=1=d$.
Both sides of the equivalence in condition~2 are
therefore false, so the equivalence holds.

Finally, the probability table uniform on $E$ is
\[
p=
\begin{pmatrix}
1/2 & 0\\
0 & 1/2
\end{pmatrix}.
\]
This table and $r=1$ have constant encoding length,
which verifies condition~3.

Suppose now that $k\le e$ and that some vertex of $\Gamma$
has degree at least $k$. Then $(\Gamma,k)$ is a yes-instance
of Maximum Edge Biclique. Indeed, if such a vertex
$x$ belongs to $L$, choose any $k$ of its neighbors
and let $B$ be their set. The sets $A=\{x\}$ and $B$
form a biclique with $|A||B|=k$ edges.
If the vertex belongs to $R$, the same argument applies
with the two sides exchanged.

In this case, construct $\Gamma'=(\mathcal X,\mathcal Y,E)$
by setting
\[
\mathcal X=\{x_1,x_2\},\qquad
\mathcal Y=\{y_1,y_2\},\qquad
E=\mathcal X\times\mathcal Y,
\]
where the four vertices are distinct, and put
$d=2$, $r=1$, and $n=2$.

For condition~1, we have $|\mathcal X|=2=2^r$,
and every left and right vertex has degree $2=d$.

For condition~2, we have already shown that $\Gamma$
contains a biclique with at least $k$ edges.
The constructed graph also contains a rectangle
with more than $d$ entries: take
$A=\mathcal X$ and $B=\mathcal Y$. Then
\[
A\times B=E,\qquad |A||B|=4>2=d.
\]
Both sides of the equivalence in condition~2 are
therefore true, so the equivalence holds.

Finally, the probability table uniform on $E$ is
\[
p=
\begin{pmatrix}
1/4 & 1/4\\
1/4 & 1/4
\end{pmatrix}.
\]
This table and $r=1$ have constant encoding length,
which verifies condition~3.

\medskip
In every remaining case, the following conditions hold
\begin{equation}
2\le k\le e,\qquad \deg_\Gamma(v)\le k-1\quad\text{for every vertex }v.
\label{eq:nontrivial}
\end{equation}
Indeed, if $k=1$ and $e>0$, the preceding yes-case applies.

\smallskip
\noindent\emph{Step 2}
Remove isolated vertices from $\Gamma$; in fact, none can belong to a biclique with both sides
nonempty. Let $L_0,R_0$ be the remaining vertex sets, put $n_0=|L_0|$,
and set $d=k-1$. Each surviving vertex has degree between $1$ and $d$.
Since $\sum_{x\in L_0}\deg_\Gamma(x)=e$, we have $n_0\le e$.
Moreover, we $n_0\ge2$: otherwise, i.e., if there were only one left vertex, its degree would be
$e\ge k$, contrary to~\eqref{eq:nontrivial}.

 Let $n=2^r$ be the smallest power of two at least $n_0$.
The integers $n$ and $r$ are computable directly from
the binary representation of $n_0$ in polynomial time.
By minimality,
\begin{equation}
2^{r-1}<n_0\le 2^r,
\qquad\text{hence}\qquad
n_0\le n<2n_0,\quad r\ge1.
\label{eq:rounding}
\end{equation}
Construct the final graph by the following two operations:
\begin{itemize}
\item For each $x\in L_0$, add $d-\deg_\Gamma(x)$ new right vertices,
each adjacent only to $x$.
\item Add $n-n_0$ new left vertices. For each such new vertex $\alpha$,
create 
$d$ new right neighbors,
each adjacent \textit{only} to that left vertex $\alpha$.
\end{itemize}
All introduced vertices are distinct, and no other edges are added.
Write $\mathcal X,\mathcal Y$ for the final vertex sets and $E$ for the final
edge set. Every left degree is now $d$. An original right degree is unchanged
and lies between $1$ and $d$; every new right degree is one. 

This proves
the first assertion  of the lemma. 
We verify both implications of the second assertion directly in the final
graph. 

Suppose first that $\Gamma$ contains a biclique $A\times B$
with at least $k$ edges. None of its vertices is isolated,
so all survive the deletion of isolated vertices.
All its edges are retained in $\Gamma'$, and therefore
$A\times B\subseteq E$ with $|A||B|\ge k=d+1>d$.

Conversely, suppose $A\times B\subseteq E$ and
$|A||B|>d$. We show that this rectangle uses only
original vertices and edges.

Firstly, let us recall that 
every newly added right vertex is \textit{private}, that is, it has exactly one left neighbor. 
If $B$ contained a newly added right vertex $y_0$,
let $x_0$ be its unique neighbor. Since every vertex
of $A$ must be adjacent to $y_0$, we would have
$A=\{x_0\}$. Then $B$ would be contained in the
$d$-element neighborhood of $x_0$, giving
$|A||B|=|B|\le d$, a contradiction.
Hence $B\subseteq R_0$. 

Secondly, every newly added left vertex is adjacent only to
newly added right vertices, so it has no neighbor in
the nonempty set $B$. Since every vertex of $A$
must be adjacent to every vertex of $B$, no newly
added left vertex can belong to $A$. Thus $A\subseteq L_0$.
Finally, the construction adds no edges between
$L_0$ and $R_0$, so
\[
A\times B\subseteq E\cap(L_0\times R_0)=E_\Gamma.
\]
This is a biclique of the original graph with
$|A||B|>d=k-1$, hence with at least $k$ edges.

\smallskip
\noindent\emph{Step 3.}
Now, we estimate  the size of the constructed instance and the running time of the algorithm.

Let $\ell$ be the total bit length of $(\Gamma,k)$. The graph is given explicitly,
so $e=O(\ell)$. In the nontrivial case,
\[
n<2n_0\le2e,\qquad 1\le d=k-1<e,\qquad N=|E|=nd\le2e^2.
\]
Every right vertex has positive degree, and the sum of the right degrees is
$N$. Hence $c:=|\mathcal Y|\le N$. The complete probability table, including
zeros, has at most
\[
nc\le nN\le4e^3
\]
cells. Its entries are $0$ or $1/N$, each requiring $O(\log(e+1))$ bits;
therefore its total encoding length is $O(e^3\log(e+1))$.
Also $r=\log_2 n\le\log_2(2e)$, so $r$ has polynomial encoding length.

Degrees and isolated vertices are found by scanning the input. The
construction adds polynomially many vertices and edges, which may be
relabeled consecutively, and writes polynomially many table entries.
By~\eqref{eq:rounding}, the power-of-two padding increases
the row count by less than a factor of two, and
$n<2n_0\le2e$. Thus the requirement $n=2^r$ preserves
the polynomial bound on the output size.
Finally, the test $k>e$ occurs before any loop whose length depends on the
value of $k$. All subsequent such loops have $k\le e$, so a large
binary-encoded threshold cannot cause a pseudopolynomial running time.
This proves the third assertion and the claimed polynomial-time construction.
\end{proof}

\begin{proof}[Proof of Theorem~\ref{thm:J-main}]
Given $(\Gamma,k)$, apply Lemma~\ref{lem:reduction}, form the uniform table $p$ on
$E$, and output $(p,r)$. The output satisfies all the stated restrictions,
and $H(X)=\log n=r$ exactly. By Lemma~\ref{lem:reduction} and
Corollary~\ref{cor:criterion},
\[
\begin{split}
\Gamma\text{ has a biclique with at least }k\text{ edges}
&\ \Longleftrightarrow\ E\text{ has a rectangle larger than }d\\
&\ \Longleftrightarrow\ J(p)<r.
\end{split}
\]
 The construction is polynomial-time and  proves NP-hardness by a many-one
reduction.

\medskip
For NP membership on the restricted class, treat it as a language on all
encodings and reject those failing its restrictions. First check that the
input specifies a probability table and a nonnegative integer $r$.
The remaining checks are also polynomial-time: count rows, check that their number is a power of
two whose exponent equals $r$, count positive entries in each row and column,
and compare each positive rational entry with $1/(nd)$. The exponent is
read from the binary row count; there is no need to construct $2^r$ for an
arbitrarily large supplied $r$.
For an input satisfying the restrictions, a certificate consists of two
nonempty index sets $A,B$. Verify that every entry in $A\times B$ is positive
and that $|A||B|>d$. This certificate has polynomial length, and
Corollary~\ref{cor:criterion} proves that such a certificate exists exactly
when $J(p)<r$. No real-valued certificate or entropy comparison is required.
Thus the restricted language is NP-complete.
\end{proof}

\begin{remark}
\label{rem:strict}
Every constructed instance satisfies $J(p)\le r$, with equality exactly in
the no case. Consequently, replacing $J(p)<r$ by $J(p)\le r$ at this same
threshold would make every constructed instance a yes-instance. Section~\ref{sec:gap} will establish a quantitative gap
between the two cases. This will allow us to replace $r$
by the smaller rational threshold
$\theta=r-1/(8N^2)$, where $N$ is the number of positive
table entries, so that yes-instances satisfy
$J(p)<\theta$ and no-instances satisfy $J(p)>\theta$.
\end{remark}

\Needspace{10\baselineskip}
\section{A quantitative gap and hardness of numerical approximation}
\label{sec:gap}

To prove $J(p)<\log n$, Lemma~\ref{lem:large-rectangle}
needed only an admissible $W_t$ for which $\Phi(t)>0$,
because
\[
J(p)\le I(X,Y;W_t)=\log n-\frac{\Phi(t)}N.
\]
We now choose an explicit rational value of $t$ and
derive a positive lower bound on $\Phi(t)/N$.
This quantifies the separation between yes-instances
and no-instances and allows them to be distinguished
using rational approximations of $J(p)$.

\subsection{An explicit separation from the row entropy}

\begin{proposition}
\label{prop:gap}
Let $p$ be the table~\eqref{eq:uniform}, with $N=nd$. If its support
contains a rectangle with more than $d$ entries, then
\begin{equation}
J(p)<\log n-\frac1{4(d+1)N}
\le\log n-\frac1{4N^2}.
\label{eq:gap}
\end{equation}
If no such rectangle exists, $J(p)=\log n$.
\end{proposition}

\begin{proof}
The equality assertion is Lemma~\ref{lem:no-rectangle}.
Otherwise choose a rectangle of size $m>d$ and use the auxiliary and gain
function of Lemma~\ref{lem:large-rectangle}. Since $m,d$ are integers,
$m\ge d+1$. Set
\[
t=\frac1{2(d+1)}\in(0,1).
\]
Applying $\ln u\le u-1$ with $u=d/(d+1)$ gives
\[
\ln\frac{d+1}{d}\ge\frac1{d+1}.
\]
Consequently,
\[
\log\frac md\ge\log\frac{d+1}{d}
\ge\frac1{(d+1)\ln2}=\frac{2t}{\ln2}.
\]
The estimate~\eqref{eq:gain-bound} now yields
\[
\Phi(t)\ge\frac{mt^2}{\ln2}
\ge\frac{(d+1)t^2}{\ln2}
=\frac1{4(d+1)\ln2}
>\frac1{4(d+1)},
\]
where the last inequality uses $0<\ln2<1$.
Since $W_t$ is admissible,~\eqref{eq:gain-identity} implies
\[
J(p)\le I(X,Y;W_t)=\log n-\frac{\Phi(t)}N
<\log n-\frac1{4(d+1)N}.
\]
Finally, the rectangle is contained in $E$, so $d+1\le m\le N$.
This proves the second inequality in~\eqref{eq:gap}.
\end{proof}

The chosen $t$ depends only on $d$ and has $O(\log(d+1))$ encoding length.
For the outputs of Lemma~\ref{lem:reduction}, the resulting dichotomy is
\begin{equation}
\begin{cases}
J(p)<r-1/(4N^2),&\text{if the biclique instance is yes},\\
J(p)=r,&\text{if the biclique instance is no}.
\end{cases}
\label{eq:dichotomy}
\end{equation}
In the nontrivial construction, $N\le2e^2$, so $1/(4N^2)\ge1/(16e^4)$;
the fixed outputs have constant $N$. The separation is therefore
inverse-polynomial in the original input length, as well as in the length
of the explicitly supplied table.

\subsection{Hardness with a finite rational output}
In this section we show the hardness of numerical approximation of the Wyner common information, within a certain threshold.
\begin{theorem}
\label{thm:approximation}
On the class of tables in Theorem~\ref{thm:J-main}, let $N$ be the number of
positive entries. Producing a rational $q$ satisfying
\begin{equation}
|q-J(p)|\le\frac1{16N^2}
\label{eq:required-error}
\end{equation}
is NP-hard under a polynomial-time reduction making one call to a solver
for this approximation task. In particular, a polynomial-time algorithm
guaranteed to return such a $q$ would imply $\PT=\NP$.

Consequently, unless $\PT=\NP$, no algorithm can, for every such $p$ and
every positive integer $s$ given in unary, return a rational $q$ with
$|q-J(p)|\le2^{-s}$ in time polynomial in the table length and $s$.
The obstruction already occurs for $s=O(\log N)$.
\end{theorem}

\begin{proof}
Given $(\Gamma,k)$, construct $(p,r)$ by Lemma~\ref{lem:reduction}, compute $N=nd$,
and obtain any rational $q$ satisfying~\eqref{eq:required-error}. Compute
\begin{equation}
\theta=r-\frac1{8N^2}.
\label{eq:theta}
\end{equation}
On a yes-instance,~\eqref{eq:dichotomy} gives
\[
q\le J(p)+\frac1{16N^2}<r-\frac3{16N^2}<\theta.
\]
On a no-instance it gives
\[
q\ge J(p)-\frac1{16N^2}=r-\frac1{16N^2}>\theta.
\]
Thus the original instance is yes exactly when $q<\theta$, for every
allowed output of the approximation solver. If the solver runs in polynomial
time, $q$ has polynomial encoding length; comparison with $\theta$ uses
integer cross-multiplication and takes polynomial time. This proves the
first assertion.

For the precision formulation, let $s$ be the number of binary digits of
the positive integer $16N^2$. Then
\[
2^{-s}<\frac1{16N^2},\qquad s=O(\log N).
\]
This choice uses only integer operations. Since $N$ is polynomially bounded
in the original input length, the unary encoding of $s$ also has polynomial
length. An algorithm meeting the stated precision specification would
therefore supply a $q$ satisfying~\eqref{eq:required-error}, contradicting
the first assertion unless $\PT=\NP$.
\end{proof}

\begin{remark}
Braun, Jain, Lee and Pokutta~\cite{BraunJainLeePokutta2017} showed that
Wyner Common Information
$J$ characterizes the logarithm of the nonnegative rank amortized under
tensor powers and small $\ell_1$ perturbations. Combined with their
results, Theorems~\ref{thm:J-main} and~\ref{thm:approximation} imply that
this amortized nonnegative rank is also hard to compute and to
approximate. A precise statement requires fixing the order of the limits
involved in its definition, and we do not pursue it here.
\end{remark}

\section{Hardness of the common entropy}\label{sec:x3c}

Sections~\ref{sec:x3c}--\ref{sec:G-proof} are devoted to the common entropy
and to the proof of Theorem~\ref{thm:G-main}. Several letters, such as $n$,
$N$, $m$, $t$, $a$, $b$, $k$, $q$ and $s$, are used in these sections with
meanings fixed below, unrelated to those they have in
Sections~\ref{sec:rectangles}--\ref{sec:gap}.

\subsection{Outline of the proof of Theorem~\ref{thm:G-main}}\label{subsec:overview}

The reduction is from exact cover by $3$-sets and proceeds in five steps.

\begin{enumerate}[label=\textup{(\arabic*)},leftmargin=*,widest=5]
\item \emph{Decompositions.} Admissible auxiliaries are the same as decompositions
$p=\sum_w F_w$ into finitely many nonnegative matrices of rank one; the distribution of $W$ is the
vector of masses of the matrices $F_w$ (Lemma~\ref{lem:decomp}). A nonnegative rank-one matrix has
rectangular support.

\item \emph{Construction.} Given $3$-element sets $S_1,\dots,S_m$ in a ground set $U$ with
$\lvert U\rvert=N=3t$, let $M$ be the matrix whose rows are indexed by $U$ and whose columns are an
extra column with all entries $\eps=1/N$, followed by the incidence vectors of $S_1,\dots,S_m$; put
$p=M/Z$ with $Z=3m+1$. We use families in which every element lies in at most three sets and two
sets share at most one element; exact cover remains $\NP$-hard for such families
(Section~\ref{subsec:restricted}).

\item \emph{A universal majorant.} Rectangularity and the two restrictions imply that every
component of a decomposition of $M$ has mass at most $3+3\eps$, of which at most $3$ lies outside
the extra column. As the extra column has total mass $N\eps=1$, the normalized masses of every
decomposition are majorized by the two-level probability vector $\qs$ with $t$ entries
$(3+3\eps)/Z$ and $m-t$ entries $3/Z$ (Lemma~\ref{lem:majorization}). An exact cover yields a
decomposition whose normalized masses are exactly $\qs$ (Lemma~\ref{lem:cover}). By the Schur
concavity of entropy, $G(p)=H(\qs)$ whenever an exact cover exists.

\item \emph{A gap.} If there is no exact cover, the total normalized mass of the $t$ heaviest
components falls short of the sum of the $t$ largest entries of $\qs$ by at least $1/(NZ)$
(Lemma~\ref{lem:deficit}). Indeed, a smaller shortfall would force each of these components to live
on the rows of a single set $S_j$ and, together, to carry more than $1-\eps$ of the extra column;
this is possible only if their $t$ sets form an exact cover. A quantitative form of the Schur
concavity of entropy (Lemma~\ref{lem:entropy}) turns this shortfall into the entropy gap
$G(p)\ge H(\qs)+\Delta$ with $\Delta=1/(4N^2Z)$ (Theorem~\ref{thm:gap}).

\item \emph{A rational threshold.} Since $1/\Delta$ is polynomial in the size of the instance, a
rational number $r$ with $H(\qs)\le r<H(\qs)+\Delta$ can be computed in polynomial time by exact
integer arithmetic (Section~\ref{sec:G-proof}). Then $G(p)\le r$ if and only if an exact cover exists.
\end{enumerate}
Step~(1) is Lemma~\ref{lem:decomp}; the rest of this section and
Sections~\ref{sec:matrix}--\ref{sec:G-proof} carry out steps (2)--(5).

The structure of this argument should be compared with that of
Sections~\ref{sec:rectangles}--\ref{sec:gap}. There, a single large rectangle
in the support of $p$ suffices to improve on the auxiliary $W=X$, and the
source problem is Maximum Edge Biclique. Here, a decomposition of small entropy
must concentrate almost all the mass of $p$ on few heavy rectangles, which
forces a whole family of rectangles to fit together exactly; this is why the
source problem is an exact cover problem.

\subsection{Restricted pairs}\label{subsec:restricted}

\begin{definition}\label{def:restricted}
Let $U$ be a finite set and $\cS=(S_1,\dots,S_m)$ an indexed family of subsets of $U$. An
\emph{exact cover} is a set $I\subseteq[m]$ such that the sets $S_j$, $j\in I$, are pairwise
disjoint and their union is $U$. The pair $(U,\cS)$ is \emph{restricted} if
\begin{enumerate}[label=\textup{(R\arabic*)}]
\item\label{R1} $\lvert S_j\rvert=3$ for every $j\in[m]$;
\item\label{R2} every element of $U$ belongs to at most three of the sets $S_j$;
\item\label{R3} $\lvert S_j\cap S_k\rvert\le1$ whenever $j\ne k$;
\item\label{R4} $\lvert U\rvert=3t$ for an integer $t$ with $1\le t<m$.
\end{enumerate}
\end{definition}

Under \ref{R1} and \ref{R4}, every exact cover consists of exactly $t$ sets. The source of our
reduction is the following problem.

\begin{definition}\label{def:cx}
\CX: given a finite set $U$ and an indexed family $(T_1,\dots,T_n)$ of $3$-element subsets of $U$
such that every element of $U$ belongs to exactly three of the sets $T_j$, decide whether an exact
cover exists.
\end{definition}

Counting incidences gives $3\lvert U\rvert=3n$, so $\lvert U\rvert=n$ for every instance of \CX.

\begin{theorem}[Moore and Robson \cite{MooreRobson2001}]\label{thm:mr}
\CX\ is $\NP$-complete.
\end{theorem}

Moore and Robson \cite[Section~3]{MooreRobson2001} state this theorem for \textsc{Cubic Planar Monotone 1-in-3
Satisfiability}, which is the planar special case of \CX\ in the language of formulas: the clauses
form the ground set, each variable $x$ gives the set $T_x$ of the three clauses containing it (the
variable--clause incidence graph is cubic, so every clause has three distinct variables and every
variable lies in three distinct clauses), and a truth assignment makes exactly one variable of every
clause true if and only if its true variables index an exact cover; membership in $\NP$ is clear.
Of the regularity of \CX, our
reduction uses only that every element lies in at most three sets, which is used to
verify~\ref{R2}, and that $\lvert U\rvert=n$, which lets us express all the sizes
in~\eqref{eq:sizes} below in terms of $n$ alone.

\begin{figure}[hbt]
\centering
\begin{tikzpicture}[x=1cm,y=0.95cm,
    elt/.style={circle,fill=black,inner sep=1.4pt},
    lab/.style={font=\small}]
  \foreach \i/\y in {1/0,2/-1,3/-2} {
    \node[elt] (t\i) at (0,\y) {};
    \node[lab,left=3pt] at (t\i) {$u_{j,\i}$};
  }
  \draw[rounded corners=7pt,thick] (-1.05,-2.42) rectangle (0.38,0.42);
  \node[lab,above] at (-0.33,0.42) {$T'_j$};
  \draw[->,thick] (1.0,-1) -- node[above,lab] {replace} (2.6,-1);
  \foreach \i/\y in {1/0,2/-1,3/-2} {
    \node[elt] (u\i) at (4.2,\y) {};
    \node[lab,left=3pt] at (u\i) {$u_{j,\i}$};
    \node[elt] at (5.6,\y) {};
    \node[elt] at (6.8,\y) {};
    \draw[rounded corners=7pt,thick,blue!70!black] (3.15,\y-0.3) rectangle (7.3,\y+0.3);
    \node[lab,blue!70!black,right] at (7.35,\y) {$A_{j,\i}$};
  }
  \foreach \k/\x in {1/5.6,2/6.8} {
    \draw[rounded corners=7pt,thick,dashed,red!70!black] (\x-0.4,-2.45) rectangle (\x+0.4,0.45);
    \node[lab,red!70!black,below] at (\x,-2.45) {$B_{j,\k}$};
  }
\end{tikzpicture}
\caption{The gadget of Lemma~\ref{lem:restricted}. The triple $T'_j$ is replaced by three rows
$A_{j,i}$ (solid) and two columns $B_{j,k}$ (dashed) of a $3\times2$ grid of new elements
$c_{j,i,k}$; the row $A_{j,i}$ also contains the original element $u_{j,i}$. Since every row meets
every column, an exact cover uses either all three rows, and thereby covers $T'_j$, or both
columns, and then covers no original element.}
\label{fig:gadget}
\end{figure}
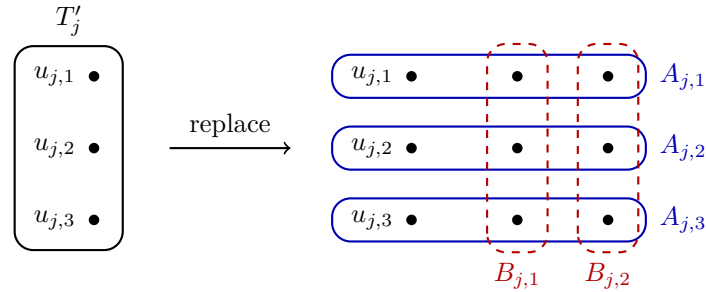

The next lemma enforces the intersection condition \ref{R3} with a simple gadget, shown in
Figure~\ref{fig:gadget}.

\begin{lemma}\label{lem:restricted}
There is a polynomial-time algorithm that maps every instance of \CX\ with $n\ge1$ sets to a
restricted pair $(U,\cS)$ with
\begin{equation}\label{eq:sizes}
  \lvert U\rvert=N=21n,\qquad t=7n,\qquad m=15n,
\end{equation}
such that $(U,\cS)$ has an exact cover if and only if the given instance has one.
\end{lemma}

\begin{proof}
Let $(U_0,(T_1,\dots,T_n))$ be the given instance of \CX, so $\lvert U_0\rvert=n$.

\emph{Three copies.} The disjoint union of three copies of the given instance is an instance of
\CX\ with a ground set $U'$ of size $3n$ and sets $T'_1,\dots,T'_{3n}$. It has an exact cover if
and only if the original instance has one, since an exact cover of the union restricts to an exact
cover of each copy, and exact covers of the copies combine to one of the union. (This step only
serves to make $\lvert U\rvert$ divisible by three in \eqref{eq:sizes}.)

\emph{Gadgets.} For each $j\in[3n]$ write $T'_j=\{u_{j,1},u_{j,2},u_{j,3}\}$, introduce six new
elements $c_{j,i,k}$ ($i\in[3]$, $k\in[2]$), the cells of a $3\times2$ grid, and replace $T'_j$ by
the three \emph{rows} and the two \emph{columns} of this grid, each row enlarged by one element of
$T'_j$:
\[
  A_{j,i}=\{u_{j,i},\,c_{j,i,1},\,c_{j,i,2}\}\quad(i\in[3]),\qquad
  B_{j,k}=\{c_{j,1,k},\,c_{j,2,k},\,c_{j,3,k}\}\quad(k\in[2]).
\]
Let $U$ consist of $U'$ and the $18n$ new elements, and let $\cS$ be the family of all $15n$ rows
and columns. Then \eqref{eq:sizes} holds, and the construction takes polynomial time.

\emph{Restrictions.} Property \ref{R1} is clear, and \ref{R4} holds because $7n<15n$. An element
$u\in U'$ belongs to the row $A_{j,i}$ exactly when $u=u_{j,i}$; so it belongs to one row of each
of the three gadgets $j$ with $u\in T'_j$, and to no other set. A cell $c_{j,i,k}$ belongs only to
$A_{j,i}$ and $B_{j,k}$. This proves \ref{R2}. As for \ref{R3}, two sets from different gadgets can
only share elements of $U'$, and each set contains at most one element of $U'$; within a gadget,
the rows are pairwise disjoint, the two columns are disjoint, and
$A_{j,i}\cap B_{j,k}=\{c_{j,i,k}\}$.

\emph{Equivalence.} Let $I$ be an exact cover of $U$ and fix $j$. Every row of gadget $j$ meets
every column of gadget $j$, so $I$ cannot contain both a row and a column of this gadget. If $I$
contains a row, it contains no column, and each cell $c_{j,i,1}$ must be covered by its row; thus
all three rows belong to $I$. If $I$ contains no row, the cells $c_{j,1,1}$ and $c_{j,1,2}$ must be
covered by the columns; thus both columns belong to $I$.

Let $J$ be the set of indices $j$ such that the rows of gadget $j$ belong to $I$. Each $u\in U'$ is
covered by exactly one member of $I$, and the members of $\cS$ containing $u$ are the rows
$A_{j,i}$ with $u_{j,i}=u$. Hence $u$ belongs to $T'_j$ for exactly one $j\in J$; that is, $J$ is an
exact cover of $U'$. Conversely, if $J$ is an exact cover of $U'$, choose the rows of every gadget
$j\in J$ and the columns of every gadget $j\notin J$. Every cell is covered exactly once, by its
row or by its column, and every $u\in U'$ is covered exactly once, by a row of the unique gadget
$j\in J$ with $u\in T'_j$. This is an exact cover of $U$.
\end{proof}

\section{The reduction matrix and its decompositions}\label{sec:matrix}

Throughout this section, $(U,\cS)$ is a restricted pair with $\lvert U\rvert=N=3t$ and
$\cS=(S_1,\dots,S_m)$. Rows of matrices are indexed by $U$ and columns by $\{0,1,\dots,m\}$; column
$0$ is the \emph{extra column}, and columns $1,\dots,m$ are the \emph{set columns}.

\subsection{The construction}

Put
\begin{equation}\label{eq:eps-Z}
  \eps=\frac1N,\qquad Z=3m+1,
\end{equation}
and define the matrices $M$ and $p$ by
\begin{equation}\label{eq:M}
  M(u,0)=\eps,\qquad M(u,j)=\begin{cases}1,&u\in S_j,\\0,&u\notin S_j,\end{cases}\quad(j\in[m]),
  \qquad p=\frac{M}{Z}.
\end{equation}
The set columns contain $3m$ ones and the extra column has mass $N\eps=1$, so $\mu(M)=Z$ and $p$
is a rational probability table. By \ref{R2}, every row of $p$ has at most four nonzero entries,
and every nonzero entry equals $1/Z$ or $1/(NZ)$.

Next, let
\begin{equation}\label{eq:qstar}
  a=\frac{3+3\eps}{Z},\qquad b=\frac3Z,\qquad
  \qs=(\underbrace{a,\dots,a}_{t},\underbrace{b,\dots,b}_{m-t}).
\end{equation}
Since $3t\eps=1$, we have $ta+(m-t)b=(3m+3t\eps)/Z=1$, so $\qs$ is a probability vector; moreover
$a>b>0$, and $1\le t<m$ by \ref{R4}. Using $ta=(N+1)/Z$ and $a/b=1+1/N$, we find
\begin{equation}\label{eq:h0}
  h_0:=H(\qs)=ta\log\frac1a+(1-ta)\log\frac1b=\log\frac Z3-\frac{N+1}{Z}\log\Bigl(1+\frac1N\Bigr).
\end{equation}
The parameters are summarized in Table~\ref{tab:parameters}.

The construction rests on the following idea. A rank-one component of mass larger than
$3+\eps$ uses exactly one set column $j$, and only rows of $S_j$; its mass is therefore at most the
$3$ units of column $j$ plus $3\eps$ units taken from the extra column. The extra column has total
mass $1=3t\eps$, just enough to give $t$ components the maximal mass $3+3\eps$, and this is possible
only if their sets form an exact cover. Lemmas~\ref{lem:component}--\ref{lem:deficit} below make
this precise.

\begin{table}[htbp]
\centering
\begin{tabular}{@{}llll@{}}
\toprule
Symbol & Definition & Value in Lemma~\ref{lem:restricted} & Role\\
\midrule
$N$ & $\lvert U\rvert$ & $21n$ & number of elements; rows of $p$\\
$t$ & $N/3$ & $7n$ & number of sets in an exact cover\\
$m$ & $\lvert\cS\rvert$ & $15n$ & number of sets; set columns of $p$\\
$\eps$ & $1/N$ & $1/(21n)$ & entries of the extra column of $M$\\
$Z$ & $3m+1$ & $45n+1$ & mass of $M$\\
$a$, $b$ & $(3+3\eps)/Z$, $3/Z$ & & the two levels of $\qs$\\
$h_0$ & $H(\qs)$, see \eqref{eq:h0} & & value of $G(p)$ on yes-instances\\
$\Delta$ & $1/(4N^2Z)$ & & entropy gap on no-instances\\
\bottomrule
\end{tabular}
\medskip
\caption{Parameters of the construction. The third column gives their values for the restricted
pairs produced by Lemma~\ref{lem:restricted} from an instance of \CX\ with $n$ sets.}
\label{tab:parameters}
\end{table}

We index the components $F_1,F_2,\dots$ of a rank-one decomposition of $M$ so that
$\mu(F_1)\ge\mu(F_2)\ge\dotsb$, and call
$q=\bigl(\mu(F_1)/Z,\mu(F_2)/Z,\dots\bigr)$ its \emph{normalized mass vector}; it is a probability
vector, listed in nonincreasing order as agreed in Section~\ref{subsec:majorization}, with $q_i=0$
for $i$ larger than the number of components.
Finally, for a matrix $F$ indexed like $M$, we write $\mu_0(F)=\sum_{u\in U}F(u,0)$ for its mass in
the extra column; thus $\mu(F)-\mu_0(F)$ is its mass in the set columns.

\subsection{Rank-one components}

The following lemma is the only place in the gap argument where the incidence restrictions
\ref{R2} and \ref{R3} are used.

\begin{lemma}\label{lem:component}
Let $F$ be a nonzero nonnegative rank-one matrix with $F\le M$. Then
\begin{enumerate}[label=\textup{(\alph*)}]
\item\label{c:a} $\mu(F)\le3+\mu_0(F)$;
\item\label{c:b} $\mu(F)\le3+3\eps$;
\item\label{c:c} if $\mu(F)>3+\eps$, then $\supp F\subseteq S_j\times\{0,j\}$ for some $j\in[m]$.
\end{enumerate}
\end{lemma}

\begin{proof}
By \eqref{eq:rectangle}, $\supp F=R\times C$ for some nonempty $R\subseteq U$ and
$C\subseteq\{0,\dots,m\}$, and $\supp F\subseteq\supp M$ because $0\le F\le M$. Let $J=C\cap[m]$ be
the set of set columns used by $F$.

If $J=\emptyset$, then $F$ is supported in the extra column, so $\mu(F)=\mu_0(F)\le\mu_0(M)=1$.

If $J=\{j\}$, then $R\times\{j\}\subseteq\supp F\subseteq\supp M$, that is, $u\in S_j$ for every
$u\in R$. Hence $\supp F\subseteq S_j\times\{0,j\}$, and $F\le M$ gives
$\mu(F)-\mu_0(F)\le\lvert S_j\rvert=3$ and $\mu_0(F)\le3\eps$.

If $J$ contains two distinct indices $j$ and $k$, then $R\subseteq S_j\cap S_k$, so $R=\{u\}$ for a
single element $u$ by \ref{R3}. By \ref{R2}, row $u$ of $M$ contains at most three ones, so
$\mu(F)-\mu_0(F)\le3$, while $\mu_0(F)\le\eps$; thus $\mu(F)\le3+\eps$.

Assertions \ref{c:a} and \ref{c:b} hold in all three cases, and $\mu(F)>3+\eps$ is possible only in
the second case, which proves \ref{c:c}.
\end{proof}

\subsection{The universal majorant and exact covers}

\begin{lemma}\label{lem:majorization}
The normalized mass vector $q$ of every rank-one decomposition of $M$ satisfies $q\preceq\qs$.
\end{lemma}

\begin{proof}
Let $F_1,F_2,\dots$ be the components, ordered as above, so that $Zq_i=\mu(F_i)$. Fix $k\ge1$; in
the following sums, $i$ ranges over the indices $i\le k$ of existing components. By
Lemma~\ref{lem:component}\ref{c:b}, $\sum_i\mu(F_i)\le k(3+3\eps)$. By
Lemma~\ref{lem:component}\ref{c:a}, $\sum_i\mu(F_i)\le3k+\sum_i\mu_0(F_i)\le3k+1$, because the
extra-column masses of all components add up to $\mu_0(M)=1$. Trivially, $\sum_i\mu(F_i)\le Z$.
Therefore
\[
  q_1+\dots+q_k\le\min\Bigl\{ka,\ \frac{3k+1}{Z},\ 1\Bigr\}.
\]
On the other hand, the partial sums of $\qs$ are $ka$ for $k\le t$, then $ta+(k-t)b=(3k+1)/Z$ for
$t\le k\le m$, and $1$ for $k\ge m$. In each range of $k$, one of the three bounds above equals the
partial sum of $\qs$.
\end{proof}

\begin{lemma}\label{lem:cover}
If $(U,\cS)$ has an exact cover, then $M$ has a rank-one decomposition with $m$ components whose
normalized mass vector is $\qs$.
\end{lemma}

\begin{proof}
Let $I\subseteq[m]$ be an exact cover, so $\lvert I\rvert=t$. For $j\in[m]$, let
$\ind_{S_j}\in\R^U$ be the indicator vector of $S_j$, let $e_0,\dots,e_m$ be the standard basis
vectors of $\R^{\{0,\dots,m\}}$, and put
\[
  E_j=\ind_{S_j}(e_j+\eps e_0)^{\top}\quad(j\in I),\qquad E_j=\ind_{S_j}e_j^{\top}\quad(j\notin I).
\]
Each $E_j$ is a nonnegative matrix of rank one. The matrices $E_j$ reproduce every set column of
$M$, and their extra columns add up to $\eps\sum_{j\in I}\ind_{S_j}=\eps\ind_U$, because the sets
$S_j$, $j\in I$, partition $U$. Thus $\sum_jE_j=M$. Finally, $\mu(E_j)=3+3\eps=Za$ for $j\in I$
and $\mu(E_j)=3=Zb$ for $j\notin I$; listed in nonincreasing order, the normalized masses form
$\qs$.
\end{proof}

In probabilistic terms, the corresponding auxiliary takes the value $j\in[m]$ with probability $a$
if $j\in I$ and $b$ otherwise; given $W=j$, the variable $X$ is uniform on $S_j$ and, independently,
$Y=j$, except that for $j\in I$ the value $Y=0$ is taken with probability $1/(N+1)$.

\subsection{Decompositions without an exact cover}

For a probability vector $q$, let
\begin{equation}\label{eq:Dt}
  D_t(q)=ta-(q_1+\dots+q_t)
\end{equation}
be the amount by which its $t$ largest entries fall short of those of $\qs$. By
Lemma~\ref{lem:majorization}, $D_t(q)\ge0$ for the normalized mass vector $q$ of every rank-one
decomposition of $M$, and the decomposition of Lemma~\ref{lem:cover} has $D_t(q)=0$. The next lemma
shows that, in the absence of an exact cover, $D_t(q)$ is bounded away from zero.

\begin{lemma}\label{lem:deficit}
If $(U,\cS)$ has no exact cover, then $D_t(q)\ge1/(NZ)$ for the normalized mass vector $q$ of every
rank-one decomposition of $M$.
\end{lemma}

\begin{proof}
Let $F_1,F_2,\dots$ be the components of a rank-one decomposition of $M$, ordered as above, and let
$q$ be its normalized mass vector, so that $\mu(F_i)=Zq_i$. Suppose that $Z\,D_t(q)<\eps$; we show
that an exact cover exists.

By Lemma~\ref{lem:component}\ref{c:b}, each of the numbers $Zq_1,\dots,Zq_t$ is at most $3+3\eps$,
and together they fall short of $t(3+3\eps)$ by $Z\,D_t(q)<\eps$. Hence each of them exceeds $3+2\eps$;
in particular, the components $F_1,\dots,F_t$ exist, and Lemma~\ref{lem:component}\ref{c:c} yields
indices $j_1,\dots,j_t\in[m]$ with $\supp F_i\subseteq S_{j_i}\times\{0,j_i\}$. Put
$V=S_{j_1}\cup\dots\cup S_{j_t}$, and compare two bounds on the extra-column mass of
$F_1,\dots,F_t$. By Lemma~\ref{lem:component}\ref{c:a},
\[
  \sum_{i=1}^{t}\mu_0(F_i)\ge\sum_{i=1}^{t}\bigl(\mu(F_i)-3\bigr)=t(3+3\eps)-Z\,D_t(q)-3t
  =1-Z\,D_t(q)>1-\eps.
\]
On the other hand, $F_i(u,0)=0$ for $u\notin V$, and $\sum_{i\le t}F_i(u,0)\le M(u,0)=\eps$ for every
$u$, because the components are nonnegative and sum to $M$. Hence
$\sum_{i\le t}\mu_0(F_i)\le\eps\lvert V\rvert$.

Combining the two bounds gives $\lvert V\rvert>(1-\eps)/\eps=N-1$, so $V=U$. Thus $U$ is the union
of the $t$ sets $S_{j_1},\dots,S_{j_t}$ of size three, while $\lvert U\rvert=3t$. These sets are
therefore pairwise disjoint (in particular, the indices $j_i$ are distinct), and
$\{j_1,\dots,j_t\}$ is an exact cover.
\end{proof}

\section{The entropy gap}\label{sec:entropy-gap}

We now turn the deficit of Lemma~\ref{lem:deficit} into an entropy gap. The tool is the following
quantitative form of the Schur concavity of entropy, for vectors majorized by a vector with two
distinct positive entries: a deficit at the position $t$, where $\qs$ drops from $a$ to $b$, forces
an explicit increase of the entropy.

\begin{lemma}\label{lem:entropy}
Let $1\le t<m$ be integers and $a>b>0$ real numbers with $ta+(m-t)b=1$, and let $\qs$ be the
probability vector with $t$ entries equal to $a$ and $m-t$ entries equal to $b$. Let $q$ be a
probability vector with $q\preceq\qs$.
\begin{enumerate}[label=\textup{(\alph*)}]
\item $H(q)\ge H(\qs)$.
\item If $D_t(q)=ta-(q_1+\dots+q_t)\ge\delta$ for some $\delta$ with $0<\delta\le(a-b)/2$, then
\begin{equation}\label{eq:entropy-gap}
  H(q)-H(\qs)\ge\delta\log\frac{a-\delta}{b+\delta}.
\end{equation}
\end{enumerate}
\end{lemma}

\begin{proof}
Part~(a) is the Schur concavity of entropy recalled in Section~\ref{subsec:majorization}.

For part~(b), let $q^\circ$ be the probability vector obtained from $\qs$ by moving mass $\delta$
from one entry equal to $a$ to one entry equal to $b$:
\[
  q^\circ=(\underbrace{a,\dots,a}_{t-1},\,a-\delta,\,b+\delta,\,\underbrace{b,\dots,b}_{m-t-1}).
\]
Since $\delta\le(a-b)/2$, we have $a\ge a-\delta\ge b+\delta\ge b$, so the entries of $q^\circ$
are listed in nonincreasing order. Its partial sums are $ka$ for $k\le t-1$, $ta-\delta$ for
$k=t$, and those of $\qs$ for $k\ge t+1$. Since $q\preceq\qs$, the partial sums of $q$ are at most
$ka$ for $k\le t-1$ and at most those of $\qs$ for $k\ge t+1$; moreover,
$q_1+\dots+q_t=ta-D_t(q)\le ta-\delta$. Hence $q\preceq q^\circ$, and Schur concavity gives
$H(q)\ge H(q^\circ)$.

It remains to compare $q^\circ$ with $\qs$. Let $g(x)=x\log(1/x)$ for $x>0$. Then
\[
  H(q^\circ)-H(\qs)=\bigl[g(b+\delta)-g(b)\bigr]-\bigl[g(a)-g(a-\delta)\bigr].
\]
The derivative $g'(x)=\log(1/x)-\log e$ is decreasing. By the mean value theorem,
$g(b+\delta)-g(b)\ge\delta\,g'(b+\delta)$ and $g(a)-g(a-\delta)\le\delta\,g'(a-\delta)$. Therefore
\[
  H(q)-H(\qs)\ \ge\ H(q^\circ)-H(\qs)\ \ge\ \delta\bigl(g'(b+\delta)-g'(a-\delta)\bigr)
  =\delta\log\frac{a-\delta}{b+\delta}.\qedhere
\]
\end{proof}

\begin{theorem}\label{thm:gap}
Let $(U,\cS)$ be restricted, let $p$ and $h_0$ be given by \eqref{eq:M} and \eqref{eq:h0}, and put
$\Delta=1/(4N^2Z)$.
\begin{enumerate}[label=\textup{(\alph*)}]
\item If $(U,\cS)$ has an exact cover, then $G(p)=h_0$, and the infimum in \eqref{eq:G} is attained
by an auxiliary with $m$ values and rational probabilities.
\item If $(U,\cS)$ has no exact cover, then $H(W)\ge h_0+\Delta$ for every admissible auxiliary $W$
for $p$; in particular, $G(p)\ge h_0+\Delta$.
\end{enumerate}
\end{theorem}

\begin{proof}
Let $W$ be an admissible auxiliary for $p$. By Lemma~\ref{lem:decomp}(a), applied to $p$ and rescaled
by the factor $Z$, the positive probabilities $\Pr(W=w)$ form the normalized mass vector $q$ of a
rank-one decomposition of $M$; in particular, $H(W)=H(q)$. By Lemma~\ref{lem:majorization},
$q\preceq\qs$, so Lemma~\ref{lem:entropy} applies.

(a) By Lemma~\ref{lem:entropy}(a), $H(W)\ge h_0$ for every admissible auxiliary $W$. Conversely,
Lemma~\ref{lem:cover} and Lemma~\ref{lem:decomp}(b) provide an admissible auxiliary with distribution
$\qs$, whose entropy is $h_0$.

(b) Here $D_t(q)\ge1/(NZ)$ by Lemma~\ref{lem:deficit}. Put $\delta=1/(NZ)$. Since $\eps=1/N$, we
have $a=(3N+3)/(NZ)$ and $b=3N/(NZ)$; hence $\delta\le(a-b)/2$, $a-\delta=(3N+2)/(NZ)$ and
$b+\delta=(3N+1)/(NZ)$. Lemma~\ref{lem:entropy}(b) and the inequality $\ln(1+x)\ge x/(1+x)$, with
$x=1/(3N+1)$, give
\[
  H(W)-h_0\ \ge\ \frac1{NZ}\log\frac{3N+2}{3N+1}\ \ge\ \frac{1}{(3N+2)\ln2\cdot NZ}
  \ >\ \frac{1}{4N^2Z}=\Delta,
\]
where the last inequality holds because $(3N+2)\ln2<\frac{7}{10}(3N+2)\le\frac72N<4N$ for $N\ge1$.
This holds for every admissible auxiliary $W$, and so $G(p)\ge h_0+\Delta$.
\end{proof}

\section{Proof of Theorem~\ref{thm:G-main}}\label{sec:G-proof}

We reduce \CX, which is $\NP$-hard by Theorem~\ref{thm:mr}, to \CE. An instance with $n=0$ sets
has the empty exact cover and is mapped to the yes-instance $\bigl((1),0\bigr)$; note that
$G\bigl((1)\bigr)=0$. Given an instance with $n\ge1$ sets, the reduction proceeds as follows.
\begin{enumerate}[label=\textup{(\roman*)},leftmargin=*]
\item Construct the restricted pair $(U,\cS)$ of Lemma~\ref{lem:restricted}, with $N=21n$, $t=7n$
and $m=15n$, and the matrix $p$ of \eqref{eq:M}; thus $Z=45n+1$.
\item Let $s=4N^2Z$, so that $\Delta=1/s$, and let $k$ be the least nonnegative integer such that
\begin{equation}\label{eq:threshold}
  2^{Zk}\,3^{Zs}\,(N+1)^{(N+1)s}\ \ge\ Z^{Zs}\,N^{(N+1)s}.
\end{equation}
\item Output $(p,r)$ with $r=k/s$.
\end{enumerate}

\emph{Correctness.} By \eqref{eq:h0}, $2^{Zh_0}=(Z/3)^Z\bigl(N/(N+1)\bigr)^{N+1}$. Hence, for an
integer $k\ge0$, inequality \eqref{eq:threshold} states that $2^{Zk}\ge2^{sZh_0}$, that is,
$k\ge sh_0$. Therefore $k=\lceil sh_0\rceil$, and
\[
  h_0\le r<h_0+\frac1s=h_0+\Delta.
\]
If the given instance has an exact cover, then so does $(U,\cS)$ by Lemma~\ref{lem:restricted}, and
Theorem~\ref{thm:gap}(a) gives $G(p)=h_0\le r$. Otherwise $(U,\cS)$ has no exact cover, and
Theorem~\ref{thm:gap}(b) gives $G(p)\ge h_0+\Delta>r$. Hence $G(p)\le r$ if and only if the given
instance has an exact cover.

\emph{Running time.} Step (i) takes polynomial time by Lemma~\ref{lem:restricted}, and $p$ has
$N(m+1)=O(n^2)$ entries, each equal to $0$, $1/Z$ or $1/(NZ)$. Since $h_0=H(\qs)\le\log m<\log Z$,
the integer $k$ lies between $0$ and $s\lceil\log Z\rceil$. As the left-hand side of
\eqref{eq:threshold} increases with $k$, the integer $k$ can be found by binary search, with
$O(\log(s\log Z))=O(\log n)$ evaluations of \eqref{eq:threshold}. Each evaluation
compares two integers with $O(sZ\log Z)=O(n^4\log n)$ binary digits, which can be computed by
repeated multiplication in polynomial time. Finally, $r=k/s$ has $O(\log n)$ bits. Thus the
reduction runs in polynomial time.

\emph{Restricted instances.} By \eqref{eq:M} and \ref{R2}, every row of $p$ has at most four nonzero
entries, and every nonzero entry equals $1/Z$ or $1/(NZ)$; the $1\times1$ matrix $(1)$ trivially
satisfies both restrictions. This completes the proof of Theorem~\ref{thm:G-main}.
\qed

\begin{remark}
Exact integer arithmetic avoids any error analysis, at the price of large (but polynomial-size)
integers. Alternatively, $r=\hat h+\Delta/2$ works for any rational $\hat h$ within $\Delta/4$ of
$h_0$, and such an $\hat h$ can be computed from rational approximations of $\log3$, $\log N$,
$\log(N+1)$ and $\log Z$ with $O(\log n)$ correct bits.
\end{remark}

The gap of Theorem~\ref{thm:gap} also yields hardness of numerical
approximation, in parallel with Theorem~\ref{thm:approximation} for Wyner's
common information.

\begin{corollary}[Hardness of approximating the common entropy]\label{cor:G-approx}
Producing, for a rational probability table $p$ with $L$ entries, a rational
number $q$ with
\[
|q-G(p)|\le\frac1{36L^2}
\]
is $\NP$-hard under a polynomial-time reduction making one call to a solver for
this task. This remains true for the restricted instances of
Theorem~\ref{thm:G-main}.
\end{corollary}

\begin{proof}
Instances of \CX\ with $n=0$ sets are yes-instances and are decided directly.
Given an instance with $n\ge1$ sets, construct $p$ as in step~(i) above; it has
$L=N(m+1)$ entries. Since $Z=3m+1\le3(m+1)$ and $N\le L$, we have
$4N^2Z\le12NL\le12L^2$, so every allowed output $q$ satisfies
$|q-G(p)|\le\Delta/3$. Let $s'=3s$, and let $k'$ be the least nonnegative
integer satisfying~\eqref{eq:threshold} with $k$ and $s$ replaced by $k'$ and
$s'$. As in the proof above, $k'=\lceil s'h_0\rceil$, so that
$h_0\le k'/s'<h_0+\Delta/3$, and $k'$ is computed in polynomial time by the same
binary search. Put $\theta=k'/s'+\Delta/3$, a rational number.

If the given instance has an exact cover, then $G(p)=h_0$ by
Theorem~\ref{thm:gap}(a), and $q\le h_0+\Delta/3\le\theta$. Otherwise
$G(p)\ge h_0+\Delta$ by Theorem~\ref{thm:gap}(b), and
$q\ge h_0+2\Delta/3>k'/s'+\Delta/3=\theta$. Hence the given instance has an
exact cover if and only if $q\le\theta$, and the comparison takes polynomial
time.
\end{proof}

\section{Comparison of the two results and open problems}\label{sec:comparison}

The two hardness proofs rest on the same correspondence between admissible
auxiliaries and nonnegative rank-one decompositions (Lemma~\ref{lem:decomp}),
but on different combinatorial phenomena. Wyner's common information decreases
as soon as some component can spread its mass uniformly over a rectangle that
is large compared with the rows of $p$; the common entropy decreases only when
the mass of $p$ can be concentrated on few heavy components. The following
example shows that the first phenomenon does not imply the second, even on the
tables of Theorem~\ref{thm:J-main}.

\begin{proposition}\label{prop:J-vs-G}
Let $\mathcal X=\{x_1,x_2\}$ and $\mathcal Y=\{y_1,y_2,y_3,y_4\}$, and let $p$
be the uniform table on
\[
E=\{x_1\}\times\{y_1,y_2,y_3\}\ \cup\ \{x_2\}\times\{y_1,y_2,y_4\}.
\]
Then $p$ satisfies the restrictions of Theorem~\ref{thm:J-main} with $n=2$,
$d=3$ and $r=1$, and
\[
J(p)\;<\;1\;=\;G(p).
\]
\end{proposition}

\begin{proof}
The table has $n=2=2^1$ rows, three positive entries in each row, all equal to
$1/6=1/(nd)$, and at most two positive entries in each column. The rectangle
$\{x_1,x_2\}\times\{y_1,y_2\}\subseteq E$ has $4>d$ entries, so
Lemma~\ref{lem:large-rectangle} gives $J(p)<\log2=1$. The auxiliary $W=X$ gives
$G(p)\le H(X)=1$.

Conversely, let $W$ be an admissible auxiliary, and let $(F_w)$ be the rank-one
decomposition of $p$ given by Lemma~\ref{lem:decomp}(a). By~\eqref{eq:rectangle},
$\supp F_w=A_w\times B_w\subseteq E$ for nonempty sets $A_w\subseteq\mathcal X$
and $B_w\subseteq\mathcal Y$. For $i=1,2,3$, let $\tau_i$ be the probability that
$A_W$ equals $\{x_1\}$, $\{x_2\}$ and $\{x_1,x_2\}$, respectively. Since
$(x_2,y_3)\notin E$, a component whose support contains $(x_1,y_3)$ has
$A_w=\{x_1\}$; hence
$1/6=p(x_1,y_3)\le\sum_{w:\,A_w=\{x_1\}}\mu(F_w)=\tau_1$. On the other hand, the
components with $A_w=\{x_1\}$ are supported in the row of $x_1$, so
$\tau_1\le\Pr(X=x_1)=1/2$. Symmetrically, using $(x_2,y_4)$, we have
$1/6\le\tau_2\le1/2$. Since $(\tau_1,\tau_2,\tau_3)$ is the distribution of a
function of $W$,
\[
H(W)\;\ge\;f(\tau_1,\tau_2):=H(\tau_1,\tau_2,1-\tau_1-\tau_2).
\]
The function $f$ is concave on the square $[1/6,1/2]^2$, so its minimum there is
attained at a vertex. Now $f(1/2,1/2)=1$,
$f(1/2,1/6)=f(1/6,1/2)=\tfrac23+\tfrac12\log3>1$, and
$f(1/6,1/6)=\log3-\tfrac13>1$. Hence $H(W)\ge1$ for every admissible $W$, and
$G(p)=1$.
\end{proof}

In the language of the decision problems, $(p,1)$ is a yes-instance of the
problem~\eqref{eq:decision} for $J$, whereas $G(p)=1$: the large rectangle that
lowers $J$ leaves $G$ at the value of the trivial auxiliary $W=X$. This is why
the reduction of Section~\ref{sec:reduction} says nothing about $G$, and a
different reduction is needed.

We conclude with some questions left open by this work.
\begin{enumerate}
\item \emph{Membership in NP.} Theorem~\ref{thm:J-main} establishes
NP-completeness of the problem for $J$ on a restricted class of tables, thanks to
the combinatorial certificates of Corollary~\ref{cor:criterion}. We do not know
whether the problem for $J$ on general tables, or \CE, belongs to NP; for \CE\ we
assert NP-hardness only.
\item \emph{Asymptotic quantities.} The exact common information rate
$\lim_{n\to\infty}G(p^{\otimes n})/n$ of~\cite{KumarLiElGamal2014}, and the
amortized nonnegative rank of~\cite{BraunJainLeePokutta2017}, are defined by
limits over tensor powers, which have exponential size. Our results do not
address the complexity of computing them.
\end{enumerate}

\providecommand{\doi}[1]{\textsc{doi:}\,\texttt{#1}}

\end{document}